\documentclass[conference,10pt]{IEEEtran}

\IEEEoverridecommandlockouts

\usepackage[T1]{fontenc}
\usepackage[utf8]{inputenc}
\usepackage[noadjust]{cite}
\usepackage{amsmath,amssymb,amsfonts,amsthm}
\usepackage{mathtools}
\usepackage{algorithm, algorithmic}
\usepackage{listings}
\usepackage{alltt}
\usepackage{graphicx}
\usepackage{lipsum}
\usepackage{comment}
\usepackage[warn]{textcomp}
\usepackage{xcolor}
\usepackage{url}
\usepackage{xspace}
\usepackage[final]{microtype}
\usepackage{alltt}
\usepackage{tcolorbox}
\usepackage{hyperref}
\usepackage{color}
\usepackage{threeparttable}
\usepackage{mdframed}
\usepackage[figuresright]{rotating}
\usepackage{makecell}
\usepackage{color,soul}
\usepackage{xcolor,colortbl}

\usepackage{multirow}

\newtheorem{corollary}{Corollary}

\newtheorem{definition}{Definition}
\newtheorem{lemma}{Lemma}
\newtheorem{theorem}{Theorem}

\usepackage{subcaption}
\DeclareCaptionSubType*[Alph]{table}
\DeclareCaptionLabelFormat{mystyle}{Table~\bothIfFirst{#1}{ }#2}
\ifCLASSINFOpdf
\else
\fi
\begin{document}

%

\newcommand{\aqed}{A-QED\xspace}
\newcommand{\aqedd}{A-QED$^2$\xspace}

\newcommand{\program}{\ensuremath{P}}
\newcommand{\pc}{\ensuremath{pc}}

\newcommand{\batchsize}{\mathit{b}}
\newcommand{\batchconcat}{\mathit{concat}}
\newcommand{\batchsplit}{\mathit{split}}

\newcommand{\acc}{\mathit{Acc}}
\newcommand{\accoutputmap}[1]{f_{#1}}

\newcommand{\state}{\ensuremath{s}}
\newcommand{\states}{\ensuremath{S}}
\newcommand{\statesctrl}{\ensuremath{\states_C}}
\newcommand{\statesctrlargs}[1]{\ensuremath{\states_{C,#1}}}
\newcommand{\statectrl}{\ensuremath{\state_c}}
\newcommand{\statesmem}{\ensuremath{\states_M}}
\newcommand{\statesmemargs}[1]{\ensuremath{\states_{M,#1}}}
\newcommand{\statemem}{\ensuremath{\state_m}}
\newcommand{\statesrel}{\ensuremath{\states_R}}
\newcommand{\statesrelargs}[1]{\ensuremath{\states_{R,#1}}}
\newcommand{\staterel}{\ensuremath{\state_r}}
\newcommand{\statesnonrel}{\ensuremath{\states_N}}
\newcommand{\statesnonrelargs}[1]{\ensuremath{\states_{N,#1}}}
\newcommand{\statenonrel}{\ensuremath{\state_n}}
\newcommand{\statesinput}{\ensuremath{\states_{\mathit{In}}}}
\newcommand{\statesinputargs}[1]{\ensuremath{\states_{\mathit{In},#1}}}
\newcommand{\stateinput}{\ensuremath{\state_{\mathit{in}}}}
\newcommand{\statesoutput}{\ensuremath{\states_{\mathit{Out}}}}
\newcommand{\statesoutputargs}[1]{\ensuremath{\states_{\mathit{Out,#1}}}}
\newcommand{\stateoutput}{\ensuremath{\state_{\mathit{out}}}}
\newcommand{\initstates}{\ensuremath{\states_\mathit{CI}}}
\newcommand{\initexecstates}{\ensuremath{\states_{\mathit{I}}}}
\newcommand{\initexecstate}{\ensuremath{\state_{\mathit{si}}}}
\newcommand{\initstatesargs}[1]{\ensuremath{\states_{I,#1}}}
\newcommand{\initexecstatesargs}[1]{\ensuremath{\states_{\mathit{SI},#1}}}
\newcommand{\finalstates}{\ensuremath{\states_F}}
\newcommand{\finalstatesargs}[1]{\ensuremath{\states_{F,#1}}}
\newcommand{\initstatectrl}{\ensuremath{\state_{c,I}}}
\newcommand{\initstatectrlargs}[1]{\ensuremath{\state_{c,I,#1}}}
\newcommand{\finalstatectrl}{\ensuremath{\state_{c,F}}}
\newcommand{\finalstatectrlargs}[1]{\ensuremath{\state_{c,F,#1}}}

\newcommand{\getctrlstate}{\ensuremath{\mathit{ctrl}}}
\newcommand{\getmemorystate}{\ensuremath{\mathit{mem}}}
\newcommand{\getoutputstate}{\ensuremath{\mathit{out}}}
\newcommand{\getinputstate}{\ensuremath{\mathit{inp}}}
\newcommand{\getrelstate}{\ensuremath{\mathit{rel}}}
\newcommand{\getnonrelstate}{\ensuremath{\mathit{nrel}}}

\newcommand{\getseqinitstates}{\ensuremath{\mathit{initsym}}}
\newcommand{\getseqfinalstates}{\ensuremath{\mathit{final}}}
\newcommand{\generatedstateseq}{\ensuremath{\mathit{StateSeq}}}
\newcommand{\generatedstateseqcomp}{\ensuremath{\mathit{StateSeqC}}}

\newcommand{\finalstateargs}[1]{\ensuremath{\state_{\mathit{F},#1}}}
\newcommand{\finalstate}{\ensuremath{\state_{\mathit{F}}}}

\newcommand{\initstate}{\ensuremath{\state_\mathit{I}}}
\newcommand{\initstateargs}[1]{\ensuremath{\state_{\mathit{I},#1}}}
\newcommand{\stateseq}{\ensuremath{\boldsymbol{\state}}}
\newcommand{\actions}{\ensuremath{A}}
\newcommand{\action}{\ensuremath{a}}
\newcommand{\actionreset}{\ensuremath{\action_{\mathit{r}}}}
\newcommand{\inputs}{\mathit{I}}
\newcommand{\extrainputs}{\mathit{G}}
\newcommand{\extraoutputs}{\mathit{G}}
\newcommand{\inputinvalid}{\ensuremath{\action_\bot}}
\newcommand{\actioninvalid}{\ensuremath{\action_\bot}}
\newcommand{\inp}{\ensuremath{\mathit{in}}}
\newcommand{\extrainp}{\ensuremath{\mathit{in_{g}}}}
\newcommand{\extrainpargs}[1]{\ensuremath{\mathit{in_{g,#1}}}}
\newcommand{\inpinvalid}{\ensuremath{\mathit{in}_{\bot}}}
\newcommand{\inptwoargs}[2]{\ensuremath{\mathit{in}_{#1,#2}}}
\newcommand{\inpseq}{\ensuremath{\boldsymbol{\inp}}}
\newcommand{\extrainpseq}{\ensuremath{\boldsymbol{\extrainp}}}
\newcommand{\inpseqargs}[1]{\ensuremath{\boldsymbol{\inp_{#1}}}}
\newcommand{\inpseqinvalid}{\ensuremath{\boldsymbol{\inp_{\bot}}}}
\newcommand{\inpbatch}{\inp^{\batchsize}}
\newcommand{\validinpseqfunc}{\ensuremath{\mathit{Cap}}}
\newcommand{\validinpseqfuncinverse}{\ensuremath{\mathit{Cap}}}
\newcommand{\capturedinptobatchindex}{\ensuremath{C_B}}
\newcommand{\validinpseq}{\ensuremath{\boldsymbol{\inpseq^v}}}
\newcommand{\validinpseqargs}[1]{\ensuremath{\boldsymbol{\inpseq_{#1}^v}}}
\newcommand{\validinp}{\ensuremath{\inp^v}}
\newcommand{\validinptwoargs}[2]{\ensuremath{\inp_{#1,#2}^v}}
\newcommand{\dummyinp}{\ensuremath{\inp_\bot}}
\newcommand{\outputs}{\ensuremath{O}}
\newcommand{\outputinvalid}{\ensuremath{\outp_\bot}}
\newcommand{\outp}{\ensuremath{o}}
\newcommand{\outpseq}{\ensuremath{\boldsymbol{\outp}}}

\newcommand{\finalstateseq}{\ensuremath{\boldsymbol{\state_F}}}
\newcommand{\initstateseq}{\ensuremath{\boldsymbol{\state_I}}}

\newcommand{\selectinpbatch}{\ensuremath{\mathit{In}_{b}}}
\newcommand{\selectoutpbatch}{\ensuremath{\mathit{Out}_{b}}}
\newcommand{\selectinpglobal}{\ensuremath{\mathit{In}_{g}}}
\newcommand{\selectoutpglobal}{\ensuremath{\mathit{Out}_{g}}}

\newcommand{\compoundinpset}{\ensuremath{I_{c}}}
\newcommand{\compoundoutpset}{\ensuremath{O_{c}}}

\newcommand{\compoundinp}{\ensuremath{in_{c}}}
\newcommand{\compoundinpargs}[1]{\ensuremath{in_{c,#1}}}
\newcommand{\compoundinpseq}{\ensuremath{\boldsymbol{\compoundinp}}}
\newcommand{\compoundinpseqprimed}{\ensuremath{\boldsymbol{\compoundinp'}}}

\newcommand{\capcompoundinp}{\ensuremath{in_{c}^{v}}}
\newcommand{\capcompoundinpargs}[1]{\ensuremath{in_{c,#1}^{v}}}
\newcommand{\capcompoundinpseq}{\ensuremath{\boldsymbol{\capcompoundinp}}}

\newcommand{\capextrainp}{\ensuremath{in_{g}^{v}}}
\newcommand{\capextrainpargs}[1]{\ensuremath{in_{g,#1}^{v}}}
\newcommand{\capextrainpseq}{\ensuremath{\boldsymbol{\capextrainp}}}

\newcommand{\compoundoutp}{\ensuremath{o_{c}}}
\newcommand{\compoundoutpargs}[1]{\ensuremath{o_{c,#1}}}
\newcommand{\compoundoutpseq}{\ensuremath{\boldsymbol{\compoundoutp}}}

\newcommand{\capcompoundoutp}{\ensuremath{o_{c}^{v}}}
\newcommand{\capcompoundoutpargs}[1]{\ensuremath{o_{c,#1}^{v}}}
\newcommand{\capcompoundoutpseq}{\ensuremath{\boldsymbol{\capcompoundoutp}}}

\newcommand{\extraoutp}{\ensuremath{o_{g}}}
\newcommand{\extraoutpargs}[1]{\ensuremath{o_{g,#1}}}
\newcommand{\extraoutpseq}{\ensuremath{\boldsymbol{\extraoutp}}}

\newcommand{\capextraoutp}{\ensuremath{o_{g}^{v}}}
\newcommand{\capextraoutpargs}[1]{\ensuremath{o_{g,#1}^{v}}}
\newcommand{\capextraoutpseq}{\ensuremath{\boldsymbol{\capextraoutp}}}

\newcommand{\validstateseqfunc}{\ensuremath{\mathit{Cap}}}
\newcommand{\validstateseq}{\ensuremath{\boldsymbol{\stateseq^v}}}
\newcommand{\validstate}{\ensuremath{\state^v}}

\newcommand{\outpseqinvalid}{\ensuremath{\boldsymbol{\outp_{\bot}}}}
\newcommand{\outpseqargs}[1]{\ensuremath{\boldsymbol{\outp_{#1}}}}
\newcommand{\validoutseqfunc}{\ensuremath{\mathit{Cap}}}
\newcommand{\validoutseq}{\ensuremath{\boldsymbol{\outpseq^v}}}
\newcommand{\validoutseqargs}[1]{\ensuremath{\boldsymbol{\outpseq_{#1}^v}}}
\newcommand{\validout}{\ensuremath{\outp^v}}
\newcommand{\validouttwoargs}[2]{\ensuremath{\outp_{#1,#2}^v}}
\newcommand{\transfunc}{\ensuremath{T}}
\newcommand{\transfuncseq}{\ensuremath{\boldsymbol{T}}}
\newcommand{\transfuncargs}[1]{\ensuremath{T_{\mathit{#1}}}}
\newcommand{\outputfunc}{\ensuremath{F}}
\newcommand{\outputfuncaux}{\ensuremath{F_\mathit{aux}}}
\newcommand{\outputfuncargs}[1]{\outputfunc_{#1}}
\newcommand{\datasetin}{\mathit{D}}
\newcommand{\datasetout}{\mathit{Data}_{\mathit{out}}}
\newcommand{\data}{d}
\newcommand{\actiondata}{\mathit{ad}}
\newcommand{\sels}{\mathit{Sel}}
\newcommand{\sel}{\mathit{sel}}
\newcommand{\timestamps}{\mathit{Time}}
\newcommand{\timestamp}{t}
\newcommand{\specstate}{\mathit{Spec}_{S}}
\newcommand{\specout}{\mathit{Spec}}
\newcommand{\inouteq}{\simeq}
\newcommand{\inready}{\mathit{rdin}}
\newcommand{\inreadyseq}{\ensuremath{\boldsymbol{\inready}}}
\newcommand{\hostready}{\mathit{rdh}}
\newcommand{\hostreadyargs}[1]{\hostready_{#1}}
\newcommand{\hostreadyseq}{\ensuremath{\boldsymbol{\hostready}}}
\newcommand{\invalid}{\epsilon}
\newcommand{\invalidseq}{\ensuremath{\boldsymbol{\epsilon}}}
\newcommand{\counter}{\mathit{cnt}}
\newcommand{\bools}{\mathit{B}}

\newcommand{\outstate}{\mathit{Out}}

\newcommand{\concatseq}{\mathrel{\boldsymbol{\cdot}}}

\newcommand{\maxlatency}{\delta_{\mathit{max}}}
\newcommand{\latencyfunc}{\delta}
\newcommand{\outdeltaexpected}{\delta}

\newcommand{\maxwaittime}{\tau}
\newcommand{\minseq}{\mathit{minseq}}

\newcommand{\transfuncmulti}{\transfunc^{*}}
\newcommand{\outputfuncmulti}{\outputfunc^{*}}

\newcommand{\rem}[3]{{\bf\textcolor{#3}{{#2}: {#1}}}}
\newcommand{\cb}[1]{{\footnotesize\rem{#1}{Clark}{magenta}}}

\definecolor{darkgreen}{rgb}{0.18,0.54,0.34}
\definecolor{maroon}{rgb}{0.64,0.16,0.16}
\definecolor{darkpink}{rgb}{0.75,0.25,0.5}

 %
\title{Scaling Up Hardware Accelerator Verification using \aqed with Functional Decomposition \thanks{\textbf{This article will appear in the proceedings of Formal Methods in Computer-Aided Design (FMCAD 2021).}}}
%
%
%

\author{\IEEEauthorblockN{Saranyu Chattopadhyay\IEEEauthorrefmark{1}, Florian Lonsing\IEEEauthorrefmark{1}, Luca Piccolboni\IEEEauthorrefmark{2}, Deepraj Soni\IEEEauthorrefmark{5}, Peng Wei\IEEEauthorrefmark{4}, Xiaofan Zhang\IEEEauthorrefmark{6}, \\ Yuan Zhou\IEEEauthorrefmark{3}, Luca Carloni\IEEEauthorrefmark{2}, Deming Chen\IEEEauthorrefmark{6}, Jason Cong\IEEEauthorrefmark{4}, Ramesh Karri\IEEEauthorrefmark{5}, Zhiru Zhang\IEEEauthorrefmark{3},  Caroline Trippel\IEEEauthorrefmark{1},\\ Clark Barrett\IEEEauthorrefmark{1}, Subhasish Mitra\IEEEauthorrefmark{1}}
\IEEEauthorblockA{\IEEEauthorrefmark{1}Stanford University, \IEEEauthorrefmark{2}Columbia University, \IEEEauthorrefmark{3}Cornell University, \IEEEauthorrefmark{4}University of California, Los Angeles,\\ \IEEEauthorrefmark{5}New York University,  \IEEEauthorrefmark{6}University of Illinois, Urbana-Champaign}}

\maketitle

\thispagestyle{plain}
\pagestyle{plain}

\begin{abstract}  
Hardware accelerators (\emph{HAs}) are essential building blocks for fast and
energy-efficient computing systems. 
\textit{Accelerator Quick Error Detection (\aqed)} 
is a recent formal technique which uses Bounded Model Checking for pre-silicon verification of HAs. \aqed checks an HA for
\textit{self-consistency}, i.e., whether identical inputs within a sequence of operations always produce the same output. Under modest assumptions, \aqed is both sound and complete. However, as is well-known, large design sizes 
significantly limit the scalability of formal verification, 
including \aqed.  We overcome this scalability challenge
through a new decomposition technique for \aqed,
called \textit{\aqed with Decomposition (\aqedd)}. \aqedd
systematically decomposes an HA into smaller, functional
sub-modules, called \textit{sub-accelerators}, which are then verified independently using \aqed. We
prove completeness of \aqedd; in particular, if the full
HA under verification contains a bug, then \aqedd ensures
detection of that bug during \aqed verification of the corresponding
sub-accelerators. Results on over 100 (buggy) versions of a wide variety of HAs with millions of logic gates demonstrate the effectiveness and practicality of \aqedd.
\end{abstract}


%
\section{Introduction}

Hardware accelerators (\textit{HAs}) are critical building blocks of energy-efficient System-on-Chip (\emph{SoC}) platforms~\cite{cong2014accelerator,carloni2016case,dally2020domain}. Unlike general-purpose processors, HAs implement a set of domain-specific functions (e.g., encryption, 3D Rendering, deep learning inference), referred to as \textit{actions} in this paper, for improved energy and throughput. 
Today's SoCs integrate
dozens of diverse 
HAs (e.g., 40+ HAs in Apple's A12 mobile SoC~\cite{hill:accelerator}).

Unfortunately, the energy and throughput improvements 
enabled by HAs come at the cost of increased design complexity.
Ensuring that a given SoC will behave correctly and reliably requires verifying each and every constituent HA.
Furthermore, HAs must achieve short design-to-deployment timelines in order to meet the needs of a wide variety of evolving applications~\cite{norrie:tpu}. 
Using conventional formal verification techniques to verify HAs faces several key challenges.
Manually crafting extensive design-specific formal properties or full abstract functional specifications can be time-consuming and error-prone~\cite{foster2015trends, huang:ila:2018}.
Moreover, scaling verification to large HAs (with millions of logic gates) is difficult or even infeasible using off-the-shelf formal tools.

A recent formal verification technique targeting HAs, \textit{Accelerator-Quick Error Detection (A-QED)}~\cite{singh:aqed:2020}, overcomes the first challenge above. A-QED is readily applicable for a popular class of HAs: \textit{loosely-coupled accelerators} (\textit{LCAs})~\cite{cota:lca,cong2012architecture} (i.e., HAs that are not integrated as part of a central processing unit (\emph{CPU}), but via an SoC’s network-on-chip or a bus) that are also \textit{non-interfering}. 
Non-interfering HAs produce the same result for a given action independent of their context within a sequence of actions (not to be confused with combinational circuits). In other words, the state of the accelerator does not affect future computations, and each computation is independent from previous computations. In contrast, computations of \emph{interfering} HAs depend on state that is the result of previous computations. A-QED uses Bounded Model Checking (\textit{BMC})~\cite{clarke2001bounded} to symbolically check sequences of actions for \textit{self-consistency}.  Specifically, it checks for \textit{functional consistency (FC)}, the property that identical inputs within a sequence of operations always produce the same outputs. It was shown that FC checks, together with \textit{response bound (\textit{RB})} checks and \textit{single-action correctness (\textit{SAC})} checks, provide a thorough verification technique for non-interfering LCAs~\cite{singh:aqed:2020}. However, despite its success in discovering bugs in moderately-sized HA designs, A-QED suffers from the scalability challenges of formal tools. For example, A-QED (backed by off-the-shelf formal verification tools) times out after 12 hours when run on NVDLA, NVIDIA's deep-learning HA~\cite{nvdla} with approximately 16 million logic gates. 

In this paper, we present a new verification approach called \textit{A-QED with Decomposition (\aqedd)} to address the scalability challenge.  First, we introduce a new, more general formal model of HA execution, which captures both interfering and non-interfering LCAs. We then show how \aqedd can \emph{decompose} a large LCA into smaller \textit{sub-accelerators} in such a way that both FC and RB checks can be directly applied to the sub-accelerators.  Unlike conventional verification approaches based on decomposition, no new properties need to be devised to apply FC and RB to the decomposed sub-accelerators.  Existing decomposition approaches can be leveraged to additionally check SAC of the sub-accelerators.
\aqedd is complementary to verification approaches that rely on design abstraction, which can be used to further improve scalability and to simplify the effort required for SAC checks on decomposed sub-accelerators.

This paper presents both a formal foundation of \aqedd and an empirical evaluation that demonstrates its bug-finding capabilities in practice. 
We prove that \aqed's completeness guarantees~\cite{singh:aqed:2020} continue to hold for \aqedd---if the full HA under verification contains a bug, then \aqedd will detect that bug.
Furthermore, we apply \aqedd to a wide variety of non-interfering LCAs (although our theoretical proofs apply to interfering LCAs as well): 109 different (buggy) versions of large open-source HAs of up to 200 million logic gates (including industrial HAs). Our empirical results focus on 
designs which are described in a high-level language (e.g., C/C++) and then translated to Register-Transfer-Level (\emph{RTL}) designs (e.g., Verilog) using High-Level Synthesis (\emph{HLS}) flows, where appropriate optimizations like pipelining and parallelism are instantiated. Such HLS-based HA design flows are becoming increasingly common in industry. 
However, \aqedd is not restricted to these specific HA design styles. Our empirical results show:

\begin{enumerate}
    \item Off-the-shelf formal tools cannot handle large HAs with millions of logic gates, even when the HAs are expressed as high-level C/C++ designs. In our experiments, A-QED verification of many such HAs times out after 12 hours or runs out of memory. 
    \item \aqedd is broadly applicable to a wide variety of HAs and detects all bugs detected by conventional simulation-based verification. For very large HAs with several million (up to over 200 million) logic gates, \aqedd detects bugs in less than 30 minutes in the worst case and in a few seconds in most cases.
    \item \aqedd is thorough -- it detected all bugs that were detected by conventional (simulation-based) verification techniques. At the same time, \aqedd improves verification effort significantly compared to simulation-based verification -- $\sim5X$ improvement on average, with $\sim9X$ improvement (one person month with \aqedd vs. 9 person months with conventional verification flows) for the large, industrial designs.
\end{enumerate}

The rest of this paper is organized as follows. Sec.~\ref{sec:related_work} presents related work. Sec.~\ref{sec:theory} presents a formal model of the accelerators targeted by \aqedd and our decomposition technique. Sec.~\ref{sec:algorithm} details the \aqedd algorithms. Results are presented in Sec.~\ref{sec:results}, and 
Sec.~\ref{sec:conclusion} concludes.



%
\section{Related Work}
\label{sec:related_work}

 Conventional formal HA verification, e.g., ~\cite{DBLP:journals/tcad/CampbellLHYGRMC19,
  DBLP:conf/fpga/ChiCCW19,uvm,choi2020flash}, requires a specification, typically in the form of manually written, design-specific properties.  These are then combined with a formal model of the design and handed to a formal tool, which attempts to prove the properties or find counter-examples.
For the verification of latency-insensitive designs, an approach was
developed to automatically derive and check properties from the RTL
synthesized in HLS
flows~\cite{DBLP:journals/corr/abs-2102-06326}. However, these derived
properties are targeted at specific types of bugs.
 
  Large design sizes have always been a challenge for formal techniques, and various approaches to this problem have been proposed.
Among techniques to improve
scalability are
abstraction~\cite{DBLP:journals/jacm/ClarkeGJLV03} and compositional reasoning
(cf.~\cite{DBLP:reference/mc/GiannakopoulouNP18}).  The former removes details of the design, gaining scalability at the cost of possible false errors.  Finding a scalable abstraction that does not generate false errors can be difficult and may be impossible in some cases.  The latter uses
\emph{assume-guarantee} reasoning (e.g., \cite{DBLP:conf/icse/GiannakopoulouPC04,DBLP:conf/tacas/CobleighGP03,DBLP:journals/fmsd/GuptaMF08,DBLP:conf/cav/JhalaM01,DBLP:conf/kbse/ChoDS13,DBLP:journals/tecs/KooM09}) and can be applied to decompose a
large HA into smaller sub-modules. Importantly, the
property $p$ of the HA to be verified must also be decomposed into properties of the
sub-modules. The properties of the sub-modules are verified individually under
certain assumptions about the behavior of the other sub-modules. If
all the properties of the sub-modules hold under the respective
assumptions, then it can be concluded that $p$ holds. 
However, finding the right properties for this decomposition can be very challenging.

Unlike for general compositional reasoning, the two main components of \aqedd (FC and RB) do not require decomposing
properties.  FC, in particular, leverages a universal \emph{self-consistency} property.
Self-consistency expresses the property
that a design is expected to produce the same outputs whenever it is
provided with the same inputs~\cite{DBLP:conf/fmcad/JonesSD96}.
In \aqedd, self-consistency is checked independently for each
sub-module (sub-accelerator in our case). Importantly,
these aspects of \aqedd do not require complex assumptions about the behavior of the
other sub-modules.

It is challenging to establish general \emph{completeness
guarantees} for conventional formal verification 
techniques~\cite{DBLP:conf/charme/KatzGG99,DBLP:conf/fmcad/Claessen07,DBLP:conf/tacas/ChocklerKV01,DBLP:conf/date/GrosseKD07,DBLP:conf/dac/ChocklerKP10},
since completeness depends on the set of properties being
checked.
Designer-guided approaches~\cite{DBLP:journals/pacmpl/ChoiVSCA17,piccolboni2019kairos}
require manual effort.
Automatic generation of   
properties is usually incomplete and depends on abstract design
descriptions~\cite{DBLP:conf/fmcad/KuhneBBB10} or
models~\cite{DBLP:conf/ddecs/SoekenKFFD11}, or analysis of simulation
traces~\cite{DBLP:journals/iet-cdt/RoginKFDR09}, which may be
difficult.
In contrast, we have general completeness results for
\aqedd.

\aqedd
builds on
\aqed~\cite{singh:aqed:2020} and leverages
BMC~\cite{DBLP:conf/tacas/BiereCCZ99,clarke2001bounded}. Similar approaches based on self-consistency have
been successfully applied to other classes of hardware designs, such as processor
verification (as \emph{symbolic quick error detection
(SQED)}~\cite{DBLP:conf/itc/LinSBM15,8355908,DBLP:conf/date/SinghDSSGFSKBEM19,DBLP:conf/iccad/LonsingGMNSSYMB19,DBLP:conf/date/FadihehUNMBSK18,DBLP:conf/date/DevarajegowdaFS20}), as well as to
hardware security~\cite{DBLP:conf/date/FadihehSBMK19,
  DBLP:conf/dac/FadihehMBMSK20,
  DBLP:conf/csfw/BartheDR04,DBLP:conf/fm/BartheCK11,DBLP:conf/uss/AlmeidaBBDE16,DBLP:conf/cav/YangVSGM18}.


\section{Formal Model and Theoretical Results}
\label{sec:theory}

In this section, we introduce a formal model for HAs, define functional consistency (\emph{FC}), single-action correctness (\emph{SAC}), and responsiveness for the model, and show how these properties provide correctness guarantees.  We then define a notion of functional composition for our model and show how the above properties can be applied in a compositional way.

Our formal model differs from the one in previous work~\cite{singh:aqed:2020} in several important ways.  It allows multiple inputs to be provided simultaneously by explicitly modeling the notion of \emph{input batches}. The HAs we consider are \emph{batch-mode accelerators} as they process input batches and produce output batches. Modeling batches is useful because it more closely matches the interfaces of real HAs. Moreover, input batches enable \emph{intra-batch checks} for FC checking, as we describe below. With intra-batch checks, only one input batch is used for FC checking. Intra-batch checks are more restricted than general FC checks. However, they are easier to set up and run in practice, and they are highly effective at finding bugs, as we demonstrate empirically.   

Our model also explicitly separates control states and memory 
states. Control states represent control-flow information such as,
e.g., program counters in HLS models of HAs. Memory states represent
all other state-holding elements, e.g., program variables.

In our model we distinguish starting and ending control states in
which inputs are provided and the computed outputs are ready,
respectively. This makes the formulation simpler and is also a better
match for HLS designs written in a high-level language, which is our
main target in the experimental evaluation. Further, our 
model enables us to formulate the notion of
\emph{strong FC}, which leads to a complete approach to bug-finding with
only two input batches.

In previous work~\cite{singh:aqed:2020}, a ready-valid protocol was
used to model input/output transactions in RTL designs.
In contrast, our focus is on HLS designs. 
Finally, we distinguish so-called \emph{relevant states},
which are parts of the state space that can affect output values.
This makes it possible to model interfering as well as non-interfering
HAs. In our experiments we focus on non-interfering HAs.

Before presenting formal definitions, we illustrate terminology
informally with an example of a non-interfering batch-mode HA as shown
in Listing~\ref{ex:acc:2} (a slightly modified excerpt of an HA
implementing AES encryption~\cite{cong:bandwidth:2017}).

Function \texttt{fun} of the HA has two sub-accelerators in
lines 8-10 and 13-14 which are identified and verified by \aqedd. Each
sub-accelerator applies a certain operation to all inputs 
in an input batch of HA. In general, the \emph{batch size} of an HA is the number of
inputs in each batch, which is~256 for this HA. 
The first sub-accelerator $\mathit{ACC_1}$ processes an input batch provided via
\texttt{data} and stores its output batch in
\texttt{buf}. The second sub-accelerator $\mathit{ACC_2}$ takes its input batch from
\texttt{buf}, where it also stores the output batch it produces. 
The control state of the HA is only implicitly represented by the
program counter when executing function \texttt{fun}. Variables
\texttt{key} and \texttt{local\_key} are global and determine the
relevant state of the HA on which the result of the encryption
operation depends. The HA is non-interfering because \texttt{key} and
\texttt{local\_key} are left unchanged by $\mathit{ACC_1}$ and $\mathit{ACC_2}$. 
Constants \texttt{BS}, \texttt{UF}, and \texttt{US} are
used in HLS to configure the generated RTL.

\begin{lstlisting}[label={ex:acc:2}, language=C++, caption={HA  Example (AES Encryption)}, breaklines=true,basicstyle=\scriptsize,     keywordstyle=\color{maroon}\bfseries, commentstyle=\color{darkgreen}\textit, stringstyle=\color{darkpink}\ttfamily, numbers=left, escapeinside={(*}{*)}, numbers=left,
xrightmargin=.3em,
xleftmargin=2.5em,frame=single,framexleftmargin=2em, numberstyle=\tiny]
#define BS ((1) << 12) // BUF_SIZE 
#define UF 2 // UNROLL_FACTOR 
#define US BS/UF // UNROLL_SIZE 

void fun(int data[BS], int buf[UF][US], int key[2]){
  int j, k;
  (*\textit{\textcolor{darkgreen}{// ===${ACC_1}$ START===}}*)
  for(j=0; j<UF; j++)  
    for(k = 0; k < BS/UF; k ++) 
      buf[j][k] = *(data + i*BS + j*US + k)^key[0];
  (*\textit{\textcolor{darkgreen}{// ===${ACC_1}$ END===}}*)
  (*\textit{\textcolor{darkgreen}{// ===${ACC_2}$ START===}}*)
    for(j=0; j<UF; j++){
      aes256_encrypt(local_key[j], buf[j]);} 
  (*\textit{\textcolor{darkgreen}{// ===${ACC_2}$ END===}}*)
}
\end{lstlisting}  

\begin{definition}
\label{def:batch:mode:acc}
A \emph{batch-mode hardware accelerator (HA)} is a finite state transition system~\cite{DBLP:journals/cacm/Keller76,DBLP:conf/sagamore/Keller74} $\acc :=
(\batchsize, \actions, \datasetin, \outputs, \states, \initstatectrl, \finalstatectrl, \states_{m,I}, \transfunc)$, where
\begin{itemize}

\item $\batchsize \in \mathbb{N}$ with $\batchsize \geq 1$ is the \emph{batch size},  

\item $\actions$ is a finite set of actions,

\item $\datasetin$ is a finite set of data values, 

\item $\outputs$ is a finite set of outputs,
  
\item $\states\! =\! \statesctrl \times \statesmem$ is the set of
  \emph{states} consisting of~\emph{control~states} $\statesctrl$ and
  \emph{memory states} $\statesmem\! =\! \statesinput\! \times\! \statesoutput\!
  \times\! \statesrel\! \times\! \statesnonrel$, where
  \begin{itemize}
  \item $\statesinput = (\actions \times \datasetin)^{\batchsize}$ are the \emph{input states}, 
  \item $\statesoutput = \outputs^{\batchsize}$ are the \emph{output states}, 
  \item $\statesrel$ are the \emph{relevant states}, and 
  \item $\statesnonrel$ are the \emph{non-relevant states}, 
  \end{itemize}

\item $\initstatectrl \in \statesctrl$ is the unique \emph{initial
  control state}, which defines the set $\initexecstates =
  \{ \initstatectrl \} \times \statesmem$ of \emph{initial
    states},
    
\item $\finalstatectrl \in \statesctrl$ is the unique \emph{final
  control state}, which defines the set $\finalstates = \{
  \finalstatectrl \} \times \statesmem$ of \emph{final states},

\item $\states_{m,I}$ is the set of allowable initial memory states, which defines the set $\initstates\! = \{ \initstatectrl \}\! \times\! \states_{m,I}$ of \emph{concrete initial
    states},

\item and $\transfunc: \states \rightarrow \states$ is the \emph{state transition function}.

\end{itemize}
\end{definition}

When referring to different HAs, e.g., $\acc_0$ and $\acc_1$,
we use subscript notation to identify their components, e.g., $\acc_0
:= (\batchsize_0, \actions_0, \datasetin_0, \outputs_0, \states_0,  \initstatectrlargs{0}, \finalstatectrlargs{0}, \states_{m,I,0}, \transfunc_0)$.

We use
$\boldsymbol{v} = \langle v_1,\ldots,v_{|\boldsymbol{v}|}
\rangle $ to denote a sequence with elements denoted $v_i$ and length
$|\boldsymbol{v}|$.  We concatenate sequences (and for simplicity of notation, 
single elements with sequences) using '$\concatseq$',
e.g., $\boldsymbol{v} = v_1 \concatseq \boldsymbol{v'}$, where
$\boldsymbol{v'} = \langle v_2,\ldots,v_{|\boldsymbol{v}|}\rangle$.
We will sometimes identify a sequence $\boldsymbol{v}$ with the corresponding tuple, and we write $v \in \boldsymbol{v}$ to denote that $v$ appears in
$\boldsymbol{v}$.  We denote the $i$-th element of a tuple $t$ as $t(i)$.

An HA $\acc$ operates on a set $\inputs^{\batchsize}$ of \emph{input
  batches}, where $\batchsize$ is the \emph{batch size} and $\inputs =
\actions \times \datasetin$.  An input batch $\inp \in
\inputs^{\batchsize}$ has $\batchsize$ \emph{batch elements}, each consisting of a pair $(\action,\data)$
containing an action $\action \in \actions$ 
to be executed and data $\data \in \datasetin$ (the data on
which action $\action$ operates). 

A state $\state \in \states$ of $\acc$ with $\state =
(\statectrl, \statemem)$ consists of a control state $\statectrl \in
\statesctrl$ and a memory state $\statemem \in \statesmem$. The
control state $\statectrl$ represents control-flow-related state (e.g.,
the program counter in an execution of a high-level model of $\acc$). In a run of $\acc$, the control state starts at a distinguished
initial state $\initstatectrl$ and ends at a distinguished final state $\finalstatectrl$.

The memory state represents all other state-holding elements of $\acc$ (including, e.g., global variables, local variables, function parameters, and memory elements).  The memory state $\statemem =
(\stateinput, \stateoutput, \staterel, \statenonrel)$ is divided into
four parts.  The first part, $\stateinput \in \statesinput$, contains the input to $\acc$.
More precisely, in a run of $\acc$, the value of $\stateinput$ in the initial state is considered
the input for that run.  Similarly, at the end of a run of $\acc$, $\stateoutput \in \statesoutput$ contains the outputs
for that run (i.e., the values computed by $\acc$ based on the inputs present at the start \nolinebreak of \nolinebreak the \nolinebreak run).

The relevant state $\staterel$ represents those state elements (other than $\stateinput$) that can influence the values of the outputs.  Any part of the state that can affect the output value in at least one execution should be included in the relevant state.  As an example of when this is needed, consider an encryption HA with actions for setting the encryption key and for encrypting data.  The internal state that stores the key is part of the relevant state because it affects the way the output is computed from the input.
The non-relevant state $\statenonrel$ is everything else.
We write
$\getctrlstate(\state)$,
$\getmemorystate(\state)$,
$\getinputstate(\state)$, 
$\getoutputstate(\state)$,
$\getrelstate(\state)$, and
$\getnonrelstate(\state)$
to denote the
components $\statectrl, \statemem, \stateinput$, $\stateoutput$, $\staterel$, and $\statenonrel$, respectively.  We overload the latter four operators to apply to memory states as well, and we lift the notation to sequences of states.

The set $\initexecstates$ of initial states contains all
states resulting from combining a memory state in $\statesmem$ with
the unique initial control state $\initstatectrl$.
The concrete initial states, $\initstates$, are a subset of $\initexecstates$, and essentially represent the reset state(s) of the HA.  They play a role in defining the \emph{reachable} states (see Definition \ref{def:reachable}, below).
The set $\finalstates$ of final states contains all
states resulting from combining a memory state in $\statesmem$ with
the unique final control state $\finalstatectrl$.
Finally, the transition function $T$ defines the successor state for any \nolinebreak given \nolinebreak state in $S$.

Given an input batch $\inp \in \inputs^{\batchsize}$, the HA 
produces an \emph{output batch} $\outp \in \outputs^{\batchsize}$ as follows.
Let $\state_0 \in \initexecstates$ be an initial state with
$\getinputstate(\state_0) = \inp$, and let $\stateseq =
\transfuncseq(\state_0) = \langle \state_1, \ldots, \state_k \rangle$
denote the sequence of $|\stateseq| = k$ \emph{successor states}
generated by the \emph{transition function} $\transfunc$, where
$\state_{i} = \transfunc(\state_{i-1})$ for $1 \leq i \leq k$, such
that $\state_k \in \finalstates$ is a final state (and no earlier states in $\stateseq$ are final states).  We also assume, without loss of generality, that $\getctrlstate(\state_i)\not=\initstatectrl$ for $i>0$.
The final state
$\state_k$ holds the
output batch $\getoutputstate(\state_k) = \outp$ with $\outp \in
\outputs^{\batchsize}$ that is produced for the input batch
$\getinputstate(\state_0) = \inp$. 
Given a sequence $\stateseq$, we write $\getseqinitstates(\stateseq)$ and
$\getseqfinalstates(\stateseq)$ to denote the subsequence of
$\stateseq$ containing all initial and final states that occur in
$\stateseq$, respectively.

Given a sequence of input batches, an HA generates a sequence
of output batches based on concatenating executions for each input batch.

\begin{definition}
  \label{def:successor:states}
Let $\inpseq$ be a sequence of inputs with $n=|\inpseq|$, and let
$\state_0 \in \initexecstates$.  Then, 
$\generatedstateseq(\inpseq, \state_0)$ denotes the sequence of \emph{successor states}
of $\state_0$ that result from executing $\inpseq$, which is defined as follows. 
\begin{itemize}

\item Let $\state_0'$ be the result of replacing $\getinputstate(\state_0)$ with $\inp_1$ in $\state_0$.  Let $\stateseq'$ = $\state_0' \concatseq \transfuncseq(\state_0')$.

\item If $|\inpseq| = 1$ then $\generatedstateseq(\inpseq, \state_0) = \stateseq'$

\item If $|\inpseq| > 1$, then

  \begin{itemize}

  \item let $\state_{\mathit{f}} = \getseqfinalstates(\stateseq')$ (which is unique),
  
  \item let $\state_{\mathit{i}} = (\initstatectrl, \getmemorystate(\state_{\mathit{f}}))$,
  
  \item let $\stateseq'' = \generatedstateseq(\langle \inp_2, \ldots, \inp_n \rangle, \state_{\mathit{i}})$.

  \item Then, $\generatedstateseq(\inpseq, \state_0) = \stateseq' \concatseq \stateseq''$.

  \end{itemize}

\end{itemize}
  
\end{definition}

In Definition~\ref{def:successor:states}, the
state $\state_{\mathit{i}}$ from which each subsequent input batch is executed
is obtained from the final state $\state_{\mathit{f}}$
produced from executing the previous input batch. 
Given an HA $\acc$, we write $\generatedstateseq(\acc,
\inpseq, \state_0)$ to explicitly refer to the successor states of
$\state_0$ generated by $\acc$. If $\acc$ is clear from the context,
we omit it.

\begin{definition}
\label{def:reachable}
A state $\state \in \states$ is \emph{reachable} if $\state \in \initstates$ or if there exists a \emph{concrete
initial state} $\state_0 \in \initstates$ and sequence $\inpseq$ of
input batches such that $\state \in \generatedstateseq(\inpseq,
\state_0)$.  A relevant state $\staterel$ is reachable if $\staterel = \getrelstate(\state)$ for some reachable state $\state$.
\end{definition}

\noindent
Note that the initial states
$\initexecstates$ are not necessarily all reachable.

Next, we define an abstract specification for an HA function.
Note that we use this to define correctness, but one of the features of \aqed is that the specification is not needed for the main verification technique.

\begin{definition}[Abstract Specification]
  \label{def:spec::NEW:MODEL}
For an HA $\acc$, let $\specout: \inputs \times \statesrel
\rightarrow \outputs$ be an \emph{abstract specification function}.
\end{definition}

Definition ~\ref{def:spec::NEW:MODEL} states that the value of an output computed by an HA is completely determined by the corresponding input and the relevant part of the memory state when the HA was started.
Note that the inclusion of the relevant memory state makes the definition general enough to
model interfering HAs.  To model
non-interfering HAs, we can either make the output
dependent on only the input batch, or require that the relevant state
does not change in state transitions.

Based on the abstract specification, we define the \emph{functional
correctness} of an HA in terms of the output batches that are
produced for given input batches as follows.

\begin{definition}[Functional Correctness]
\label{def:corr:NEW:MODEL}
An HA $\acc$ is \emph{functionally correct} with respect to
an abstract specification $\specout$ if, for all concrete initial states
$\state_0 \in \initstates$ and all sequences $\inpseq$ of input
batches, if
\begin{itemize}
\item $\inpseq = \langle \inp_1, \ldots, \inp_n \rangle$,
  
\item $\stateseq = \generatedstateseq(\inpseq, \state_0)$, 

\item $\initstateseq = \getseqinitstates(\stateseq) = \langle \initstateargs{1}, \ldots, \initstateargs{n} \rangle$,
  
\item $\outpseq = \getoutputstate(\getseqfinalstates(\stateseq)) = \langle \outp_1, \ldots, \outp_n \rangle$,
\end{itemize}
then
$\forall\, j \in [1\dots b].\:  
  \outp_n(j) = \specout(\inp_n(j), \getrelstate(\initstateargs{n}))$.
\end{definition}

\noindent
A bug is simply a failure of functional correctness.

As mentioned above, even without a formal specification, we can apply the core technique of \aqed.  To do so, we leverage the concept of \emph{functional consistency}, the notion that under modest assumptions, two identical inputs will always produce the same outputs.

\begin{definition}[Functional Consistency (\emph{FC})]
\label{def:weak:func:cons:NEW:MODEL}
An HA $\acc$ is \emph{functionally consistent} if, for all
concrete initial states $\state_0 \in \initstates$ and for all sequences
$\inpseq$ of input batches, if 

\begin{itemize}

\item $\inpseq = \langle \inp_1, \ldots, \inp_n \rangle$, $\stateseq = \generatedstateseq(\inpseq, \state_0)$, 

\item $\initstateseq = \getseqinitstates(\stateseq) = \langle \initstateargs{1}, \ldots, \initstateargs{n} \rangle$,

\item $\outpseq = \getoutputstate(\getseqfinalstates(\stateseq)) = \langle \outp_1, \ldots, \outp_n \rangle$,
  
\end{itemize}
\hspace*{-.75em}
\begin{tabular}{ll}
then$\!\!\!\!$ & 
 $\forall\, i \in [1,n], j,j' \in [1,b].$ \\
  & $\inp_i(j)\! =\! \inp_{n}(j')\! \wedge\! \getrelstate(\initstateargs{i})\! =\!
  \getrelstate(\initstateargs{n}) \to  \outp_i(j)\! =\! \outp_{n}(j')$.
\end{tabular}
\end{definition}

\noindent
Definition~\ref{def:weak:func:cons:NEW:MODEL} illustrates the need for the \emph{relevant} designation for memory states.  It essentially says that two inputs, even if started at different times and in different batch positions, should produce the same output, as long as the relevant part of the memory is the same when the two inputs are sent in.  The following lemma is straightforward (see the appendix for proofs of this and other results).

\begin{lemma}[Soundness of FC]
\label{lem:sound:NEW:MODEL}
If an HA is functionally correct, then it is functionally consistent.
\end{lemma}

Checking FC requires running BMC over multiple iterations of the HA and may be computationally prohibitive for large designs or for large values of $n$.  Often, it is possible to verify a stronger property, which only requires checking consistency across two runs of the HA.

\begin{definition}[Strong FC]
\label{def:strong:fc}
An HA $\acc$ is \emph{strongly functionally consistent} if,
for all reachable initial states $\state_0, \state_0'$ and 
input batches $\inp, \inp'$, if 
\begin{itemize}
  
\item $\stateseq = \generatedstateseq(\langle \inp \rangle,
  \state_0)$, $\stateseq' = \generatedstateseq(\langle \inp' \rangle,
  \state_0')$,
  
\item $\finalstateseq = \getseqfinalstates(\stateseq) = \langle
  \finalstate \rangle$, $\finalstateseq' =
  \getseqfinalstates(\stateseq') = \langle \finalstate' \rangle$,

\item $\outpseq = \getoutputstate(\finalstateseq) = \langle \outp
  \rangle$, $\outpseq' = \getoutputstate(\finalstateseq') = \langle
  \outp' \rangle$,
  
\end{itemize}
\hspace*{-.75em}
\begin{tabular}{ll}
then$\!\!\!\!$ & 
$\forall\, j,j' \in [1,b].$ \\ 
 & $\inp(j) = \inp'(j') \wedge \getrelstate(\state_0) =
  \getrelstate(\state_0') \to  \outp(j) = \outp'(j')$.
\end{tabular}

\end{definition}

\noindent
The main difference between FC and strong FC is that the initial states $\state_0$ and
$\state_0'$ can be any reachable states.  In contrast to that, the initial state $\state_0 \in \initstates$ in the definition of FC is a concrete one. It is easy to see that strong FC implies FC, but the reverse is not true in general. This is because it may not be possible for two reachable initial states $\state_0$ and $\state_0'$ chosen in a strong FC check to both appear in a single sequence of states resulting from executing a sequence of input batches starting in a concrete initial state. Similar to previous work on \aqed for non-batch-mode HAs~\cite{singh:aqed:2020}, FC checking relies on sequences of input batches to reach all reachable states from a concrete initial state. For strong FC checking, on the other hand, two individual input batches are sufficient because the two initial states $\state_0$ and $\state_0'$ can be arbitrarily chosen from the reachable states. Like FC, strong FC is a sound approach.

\begin{lemma}[Soundness of Strong FC]
\label{lem:sound:NEW:MODEL:strong}
If an HA is functionally correct then it is strongly functionally consistent.
\end{lemma}

A challenge with using strong FC is that it requires starting
with reachable initial states.  However, we found that in practice
(cf., Section~\ref{sec:results}), it is seldom necessary to add any
constraints on the initial states.  This may seem surprising given the
well-known problem of spurious counterexamples that arises when using
formal to prove functional correctness without properly constraining
initial states.  There are at least two reasons for this.  First, many
HAs have less dependence on internal state (none for non-interfering
HAs) than other kinds of designs.  But second, and more importantly,
FC is a much more forgiving property than design-specific correctness.
Many designs are functionally consistent, even when run from
unreachable states.  In fact, we believe that this is a natural
outcome of good design and that designing for FC is a sweet spot in
the trade-off between design for verification and other design goals.
If designers take care to ensure FC, even from unreachable states,
then strong FC is both sound and easy to formulate.

Even simpler versions of the checks above can be obtained by making them \emph{intra-batch} checks.  An HA is \emph{intra-batch functionally consistent} if it is functionally consistent when $i=n=1$.  That is, intra-batch FC checks are based on sending a single input batch to the HA. Consequently, it is not necessary to identify and compare the relevant parts of the initial states (cf.~Definition~\ref{def:weak:func:cons:NEW:MODEL}) as there is precisely one initial state being used. Similarly, an HA is \emph{intra-batch strongly functionally consistent} if it is strongly functionally consistent when $s_0=s_0'$ and $\inp = \inp'$. Again, only one input batch is sent to the HA and the relevant parts of the initial states are thus always equal.  As we will show in Section~\ref{sec:results}, intra-batch checks can be a very effective approach for cheaply finding bugs. Intra-batch checks are applicable only to batch-mode HAs; i.e., they are not applicable in the context of \aqed targeted at HAs processing sequences of single inputs~\cite{singh:aqed:2020} rather than input batches.

While functional consistency alone can find many bugs, it becomes a complete technique (i.e., it finds all bugs) by combining it with \emph{single-action checks}.

\begin{definition}[Single-Action Correctness (\emph{SAC})]
\label{def:single:act:corr}
An HA $\acc$ is \emph{single-action correct (\emph{SAC})} with respect to
an abstract specification $\specout$ if, for every batch element $(\action,\data)$ and for every reachable relevant state $\staterel$, there exists some reachable initial state $\state$, such that
$\getinputstate(\state)(j) = (\action, \data)$ for some $j$,
$\getrelstate(\state) = \staterel$, and
$\getoutputstate(\getseqfinalstates(\transfuncseq(\state)))(j) = \specout((\action,\data),\staterel)$.

\end{definition}

\noindent
Essentially, SAC requires that for each action $\action$, data
$\data$, and reachable relevant state $\staterel$, we have checked
that the result is computed correctly when starting from some reachable initial state $\state$ whose relevant state matches $\staterel$. For every batch element
$(\action,\data)$ and $\staterel$, it is sufficient to run a single check where we can choose
$(\action,\data)$ to be at any arbitrary position $j$ in the
batch $\getinputstate(\state)$. Checking SAC \emph{does} require using
the specification explicitly, but these kinds of checks typically
already exist in unit or regression tests.  SAC may even be possible
to verify using simulation.  As we show in
Section~\ref{sec:results}, many bugs can be discovered without
checking SAC at all.

When formalizing single-action checks, we again advocate using an over-approximation for reachability and encourage the design of HAs with simple over-approximations for the set of reachable relevant states.  For the encryption example we gave above, the set of reachable relevant states is just the set of valid keys, which should be easy to specify.

In earlier work, using a slightly different HA model, we showed that SAC and functional consistency ensure correctness only when the HA is \emph{strongly connected (SC)}, that is, when there exists a sequence of state transitions from every reachable state to every other reachable state.  The same is true here.

\begin{lemma}[Completeness of SAC + FC + SC]
\label{lem:strong:FC:complete:NEW:MODEL:connect}
If an HA is strongly connected and single-action correct and has a bug, then it is
not functionally consistent.
\end{lemma}

\noindent
However, strong functional consistency leads to an even stronger result.

\begin{lemma}[Completeness of SAC + Strong FC]
\label{lem:strong:FC:complete:NEW:MODEL}
If an HA is single-action correct and has a bug, then it is
not strongly functionally consistent.
\end{lemma}

Finally, to address timeliness of results in addition to correctness, we define a notion of \emph{responsiveness} for our model.

\begin{definition}[Responsiveness]
\label{def:rb}
An HA is \emph{responsive with respect to bound $n$} if, for all concrete initial states $\state_0 \in \initstates$, sequences $\inpseq$ of input batches, and input batches $\inp$, if 

\begin{itemize}
  
\item $\stateseq = \generatedstateseq(\inpseq, \state_0) = \langle \state_0, \ldots, \state_{m} \rangle$ and 

\item $\stateseq' = \generatedstateseq(\inpseq \concatseq \inp, \state_0) = \langle \state_0, \ldots, \state_{m+l} \rangle$,

\end{itemize}
then $l \le n$.
\end{definition}


\subsection{Decomposition for FC Checking}

We now show how FC of a decomposed design can be derived from FC of its parts.  We first give conditions under which two HAs can be composed.

\begin{definition}[Functionally Composable]
  \label{def:composable}
$\acc_1$ and $\acc_2$ are \emph{functionally composable} if: (i) $\batchsize_1=\batchsize_2$; (ii) $\outputs_1= \actions_2 \times \datasetin_2$; (iii) $\statesctrlargs{1} \cap \statesctrlargs{2} = \emptyset$; (iv) $\statesrelargs{1}= \statesrelargs{2}$; and (v) $\statesnonrelargs{1} = \statesoutputargs{2} \times \statesnonrel'$ and $\statesnonrelargs{2} = \statesinputargs{1} \times \statesnonrel'$ for some $\statesnonrel'$.

\end{definition}

\noindent
Note in particular that composability requires that the outputs of $\acc_1$ match the inputs of $\acc_2$.  We also require that the two HAs have isomorphic memory states, which is ensured by including $\statesoutputargs{2}$ in the non-relevant states of $\acc_1$ and $\statesinputargs{1}$ in the non-relevant states of $\acc_2$.  In order to map a memory state of $\acc_1$ to the corresponding memory state in $\acc_2$, we define a mapping function $\alpha: \statesmemargs{1} \to \statesmemargs{2}$ as follows: $\alpha(\state_m) = (\getoutputstate(\state_{m}),\getnonrelstate(\state_m)(1),\getrelstate(\state_m),(\getinputstate(\state_m),\getnonrelstate(\state_m)(2)))$.  We next define functional composition.

\begin{definition}[Functional Composition, Sub-Accelerators]
\label{def:acc:decomp2}
Given functionally composable HAs $\acc_1$ and $\acc_2$, we define the \emph{functional composition}
$\acc_0 = \acc_2 \circ \acc_1$ ($\acc_1$
and $\acc_2$ are called \emph{sub-accelerators} of $\acc_0$) as follows: $\batchsize_0 = \batchsize_1$, $\actions_0 = \actions_1$, $\datasetin_0 = \datasetin_1$, $\outputs_0 = \outputs_2$, $\statesctrlargs{0} = \statesctrlargs{1} \cup \statesctrlargs{2}$, $\statesmemargs{0} = \statesmemargs{1}$, $\initstatectrlargs{0} = \initstatectrlargs{1}$, $\finalstatectrlargs{0} = \finalstatectrlargs{2}$, $\states_{m,I,0} = \states_{m,I,1}$.  The transition function is defined as follows.  $\transfunc_0(\statectrl,\statemem) =$
\begin{itemize}
    \item[(i)] if $\statectrl\in\statesctrlargs{1}$ and $\statectrl \not= \finalstatectrlargs{1}$ then $\transfunc_1(\statectrl,\statemem)$;
    \item[(ii)] if $\statectrl\in\statesctrlargs{2}$ then $\transfunc_2(\statectrl,\alpha(\statemem))$; and
    \item[(iii)] if $\statectrl = \finalstatectrlargs{1}$ then
    $(\initstatectrlargs{2},\alpha(\statemem))$.
\end{itemize}

\end{definition}

\noindent
Definition~\ref{def:acc:decomp2} essentially states that an execution of $\acc_0 = \acc_2 \circ \acc_1$ is obtained by first running $\acc_1$ to completion, then passing the outputs of $\acc_1$ to the inputs of $\acc_2$, and then running $\acc_2$ to completion.
As a variant of Definition~\ref{def:acc:decomp2}, it is also possible to
define functional composition where the sub-accelerators operate in
parallel. This way, the sub-accelerators process non-overlapping parts
of a given input batch and produce the respective non-overlapping parts of
the output \nolinebreak batch.

We now introduce a compositional version of FC.

\begin{definition}[Strong FC for Decomposition (\emph{FCD})]
\label{def:strong:fcd}
An HA $\acc$ is \emph{strongly functionally consistent for
  decomposition (strongly FCD)} if it is strongly functionally consistent and, in
addition to $\outp(j) = \outp'(j')$, the property
$\getrelstate(\finalstate) = \getrelstate(\finalstate')$ holds in the
conclusion of the implication in Definition~\ref{def:strong:fc}.
\end{definition}

\noindent
Note that strong FCD is stronger than strong FC.  In order to stitch together results on sub-accelerators, we need to establish that not only the output but also the relevant memory state is the same after processing identical inputs.  The following is clear from the definition.

\begin{corollary}
\label{cor:strong:fcd:fc}
If an HA $\acc$ is strongly FCD, then $\acc$ is strongly
FC. 
\end{corollary}

We now show that composition preserves strong FCD and then state our main result.

\begin{lemma}[Functional Composition and Strong FCD]
\label{prop:fc:decomp}
Let $\acc_0 = \acc_2 \circ \acc_1$. 
If both $\acc_1$
and $\acc_2$ are strongly FCD then $\acc_0$ is
strongly FCD.
\end{lemma}

\begin{theorem}[Completeness of \aqedd]
\label{thm:fcd:sac:completeness}
Let $\acc_0, \acc_1$, and $\acc_2$ be HAs such that $\acc_0 =
\acc_2 \circ \acc_1$ and $\acc_0$ is single-action correct.
If $\acc_1$ and $\acc_2$ are strongly FCD then $\acc_0$ is
functionally correct.
\end{theorem}

\noindent
Theorem~\ref{thm:fcd:sac:completeness} states that \aqedd is
complete. That is, by contraposition, if an HA $\acc_0$ has a
bug, i.e., it is not functionally correct, then either $\acc_1$ or
$\acc_2$ is not strongly FCD, and thus the bug can be detected by \aqedd.

Note that there is no corresponding soundness result.  This is because it is possible to decompose a functionally consistent HA into functionally inconsistent sub-accelerators.  However, as shown in Section~\ref{sec:results}, this appears to be rare in practice, and here again we reiterate our position on design for verification and advocate that also sub-accelerators should be designed with functional consistency in mind.

Functional composition can easily be generalized to more than
two sub-accelerators. Moreover, it can be applied recursively to
further decompose sub-accelerators. 
If functional decomposition based on Definition~\ref{def:acc:decomp2} is not
applicable to further decompose a sub-accelerator, then such a
sub-accelerator can be decomposed using existing formal decomposition
approaches, though these require significant manual effort.
Our approach identifies conditions under which simple, automatable decomposition of FC checking is possible.


%
\section{\aqedd Functional Decomposition in Practice
}
\label{sec:algorithm}

We now present our implementation of \aqedd, which builds on the theoretical framework of the previous section.  We combine functional decomposition with checks for FC (dFC), SAC (dSAC), and
responsiveness (dRB).

\subsection{Decomposition for FC: dFC}
\label{sec:dFC}

dFC takes as input a non-interfering LCA design $\acc$ (satisfying Definitions~\ref{def:batch:mode:acc}
and~\ref{def:successor:states}) together with designer-provided
annotations (explained in this section). dFC decomposes $\acc$ into
sub-accelerators (following Definition~\ref{def:acc:decomp2}). 
FC checks are run on the sub-accelerators and any
counterexamples are reported.
Note that the way in which
$\acc$ is actually decomposed into sub-accelerators has no influence on the
completeness of \aqedd (Theorem~\ref{thm:fcd:sac:completeness}). That said, FC checks may scale better for certain
decompositions. While failing FC checks expose consistency issues at the sub-accelerator level, it is possible that they do not cause incorrect behaviors at the full $\acc$ level. However, we did not observe any instances of this in our experiments.

Our dFC implementation relies on identifying \emph{batch operations} in a given $\acc$. A batch operation operates on a vector of inputs, applying some action to each input in order to produce a vector of outputs.  The input to a batch operation could be an intermediate output batch of another sub-accelerator or an input batch to $\acc$ itself. A batch operation produces either an intermediate output batch which is subsequently processed
by another sub-accelerator or an output batch of $\acc$ itself. 

We assume that $\acc$ is expressed in a
high-level language, specifically as a C/C++ program\footnote{HAs expressed in Verilog or SystemC can be converted
  into C/C++, and then our dFC implementation can be
  applied. We do this in Sec.~\ref{sec:results}.} that implements
sequential computation of $\acc$ outputs from $\acc$ inputs.\footnote{Existing HLS tools (e.g., Xilinx
Vivado HLS, Mentor Catapult HLS) can then optimize $\acc$, incorporate
appropriate pipelining and parallelism,
and
produce Verilog for subsequent logic synthesis and physical design steps. Such
HLS-based HA design flows are becoming increasingly common.}
Batch operations in the C/C++ program are identified by finding
contiguous C/C++
statements 
called \emph{functional blocks} that implement those batch operations. Each functional
block represents a sub-accelerator.

We have developed a set of annotations by which the designer can help
identify these functional blocks. Examples of such annotations
are given in Listing~\ref{ex:acc:1} (extends Listing~\ref{ex:acc:2}). It has two functional blocks corresponding to batch operations: lines 15-17 and 32-33.

Annotations are defined by
particular keywords that are prefixed by ``\%'' (and denoted in blue) in Listing~\ref{ex:acc:1}. These annotations describe the compute and memory access patterns of the functional block as it transforms an input batch into an output batch.  In practice, hardware
designers already use similar annotations frequently, e.g., to express
parallelization opportunities for HLS to generate efficient hardware.  As a result, we expect manageable effort in creating such
annotations to support dFC.  The HLS research community is actively
developing new techniques to automatically explore the HA design space
and derive optimal design points together with appropriate
parallelization and pipelining
~\cite{wang2017flexcl,zhao2017comba,zhong2016lin}.  With tight
integration of \aqedd with HLS, we expect that it will be possible to generate dFC
annotations with low effort.

\begin{lstlisting}[label={ex:acc:1}, language=C++, caption={C/C++ Annotation Example (AES Encryption)}, breaklines=true,basicstyle=\scriptsize,     keywordstyle=\color{maroon}\bfseries, commentstyle=\color{darkgreen}\textit, stringstyle=\color{darkpink}\ttfamily, numbers=left, escapeinside={(*}{*)}, numbers=left,
xrightmargin=.3em,
xleftmargin=2.5em,frame=single,framexleftmargin=2em, numberstyle=\tiny]
#define BS ((1) << 12) // BUF_SIZE 
#define UF 2 // UNROLL_FACTOR 
#define US BS/UF // UNROLL_SIZE 

void fun(int data[BS], int buf[UF][US], int key[2]){
  int j, k;

  (*\textcolor{blue}{\bfseries \%IN\_SIZE}*) 16 // variables per input batch element
  (*\textcolor{blue}{\bfseries \%IN\_BATCH\_SIZE}*) BS/IN_SIZE // input batch size
  (*\textcolor{blue}{\bfseries \%BATCH\_MEM\_IN}*) data // input batch source
  (*\textcolor{blue}{\bfseries \%IN\_ALLOC\_RULE}*) in(x) addr range =
    [i*BS +  x*IN_SIZE :
    i*BS + (x + 1)*IN_SIZE]  // BATCH_MEM_IN layout
  (*\textit{\textcolor{darkgreen}{// ===${ACC_1}$ START===}}*)
  for(j=0; j<UF; j++)  
    for(k = 0; k < BS/UF; k ++) 
      buf[j][k] = *(data + i*BS + j*US + k)^key[0];
  (*\textit{\textcolor{darkgreen}{// ===${ACC_1}$ END===}}*)
  (*\textcolor{blue}{\bfseries \%OUT\_SIZE}*) 16 // variables per output batch element
  (*\textcolor{blue}{\bfseries \%OUT\_BATCH\_SIZE}*) BS/OUT_SIZE // output batch size
  (*\textcolor{blue}{\bfseries \%BATCH\_MEM\_OUT}*) buf // output batch source
  (*\textcolor{blue}{\bfseries \%IN\_ALLOC\_RULE}*) out(x) addr range = 
    [x/US][(x%US)*OUT_SIZE :
    ((x + 1)%US)*OUT_SIZE] // BATCH_MEM_OUT layout

  (*\textcolor{blue}{\bfseries \%IN\_SIZE}*) 16 
  (*\textcolor{blue}{\bfseries \%IN\_BATCH\_SIZE}*) BS/IN_SIZE                  
  (*\textcolor{blue}{\bfseries \%BATCH\_MEM\_IN}*) buf
  (*\textcolor{blue}{\bfseries \%IN\_ALLOC\_RULE}*) in(x) addr range =
    [(x%US)*IN_SIZE : ((x+1)%US)*IN_SIZE][x/US] 
  (*\textit{\textcolor{darkgreen}{// ===${ACC_2}$ START===}}*)
    for(j=0; j<UF; j++){
      aes256_encrypt(local_key[j], buf[j]);} 
  (*\textit{\textcolor{darkgreen}{// ===${ACC_2}$ END===}}*)
  (*\textcolor{blue}{\bfseries \%OUT\_SIZE}*) 16 
  (*\textcolor{blue}{\bfseries \%OUT\_BATCH\_SIZE}*) BS/OUT_SIZE
  (*\textcolor{blue}{\bfseries \%BATCH\_MEM\_OUT}*) buf
  (*\textcolor{blue}{\bfseries \%OUT\_ALLOC\_RULE}*) out(x) addr range =
    [(x%US)*OUT_SIZE : ((x+1)%US)*OUT_SIZE][x/US]
}
\end{lstlisting}

From the annotations, we create sub-accelerators. For example, the annotations in Listing~\ref{ex:acc:1} generate two sub-accelerators: $\acc_1$ corresponding to the functional block in Lines 15-17 with annotations in Lines 8-13 and 19-24, and $\acc_2$ corresponding to the functional block in Lines 32-33 with annotations in Lines 26-30 and 35-39.
For each sub-accelerator, we create an \emph{\aqedd module} for FC checking.\footnote{See the appendix for details.} It generates symbolic inputs for the sub-accelerator and symbolically executes the corresponding functional block in order to produce symbolic expressions for the outputs.
For strong FC checks (Definitions~\ref{def:weak:func:cons:NEW:MODEL}
and~\ref{def:strong:fc}), the relevant states (Definition~\ref{def:batch:mode:acc}) must additionally be identified and explicitly
constrained to be consistent across sub-accelerator calls processing two input batches. Identifying the relevant states is not necessary for intra-batch FC checks (discussed in the context of Lemma~\ref{lem:sound:NEW:MODEL:strong}).  For example, in sub-accelerator
$Acc_1$ in Listing~\ref{ex:acc:1}, \textit{key[0]} is a relevant state element 
(distinct from the batch input \textit{data}).
Between two calls of $Acc_1$ during a strong FC check,
\textit{key[0]} must be consistent.
In our implementation, we ignore reachability and allow all checks to start from fully symbolic initial states.  This does not lead to spurious counterexamples in our experiments.

\subsection{Decomposition for RB: dRB}
\label{sec:dRB}

The sub-accelerators for \aqedd's RB checks (Definition~\ref{def:rb}) can be (and often are) different from those for FC because RB involves a
much simpler check: \emph{some} output is produced within the response bound $n$. We expect $n$ to be provided by the designer for the top-level accelerator.  We then use the same bound $n$ for each sub-accelerator. The rationale is that if a sub-accelerator fails
an RB check, then the full accelerator would also fail the same RB check.

For dRB, we generate a static single assignment (SSA) representation of the design.
We then apply a \emph{sliding window algorithm} to
dynamically generate sub-accelerators. 
Lines of code in the SSA that fall within a certain \emph{window
  $W$} form the sub-accelerator. Due to SSA form, the inputs of this sub-accelerator are variables that are never updated or assigned in $W$ while the outputs are the variables which update variables outside $W$. The current size of $W$ is given by the number of LOCs that fit in $W$, and it changes dynamically during a run of the algorithm to incorporate the largest sub-accelerator that will fit the BMC tool. Once the sub-accelerator is verified, $W$ slides by $\delta$ LOCs ($\delta$ is a parameter) and adjusts its boundary to get the next largest sub-accelerator that can be verified. We synthesize that 
sub-accelerator using HLS (since some responsiveness bugs only manifest after HLS) and then run RB checks using BMC.  The initial states of each generated sub-accelerator are left unconstrained (i.e., fully symbolic) in order to analyze all possible behaviors. The specific size of
$W$ and its position in the SSA code change dynamically as dRB proceeds.  dRB terminates when $W$ reaches the end of the SSA code or if at any time an RB check fails.

\subsection{Decomposition for SAC: dSAC}
\label{sec:dSAC}

As mentioned above, and as will be shown in the next section, many bugs can be detected using only dFC and dRB.  The advantage of this is that both of these checks can be run without any functional specification.  dSAC completes the story, but at the cost of requiring specifications.  We use standard functional decomposition techniques (essentially, writing preconditions, invariants, and postconditions) to decompose SAC checks.  One feature of dSAC is that only a single input in a batch needs be checked---all other inputs in the batch can be set to constants (we use zero in our experiments).  This makes both writing the properties and checking them much simpler.  The non-input part of the initial state for each check is again kept fully symbolic for simplicity.
If a sub-accelerator is too big, we further decompose it using finer-grained functional blocks.


%
\section{Experimental Results}
\label{sec:results}
\begin{table*}
\centering
\setlength{\tabcolsep}{1.5pt}
\begin{tabular}{|l c c|c|c|c|c|c|c|c|c|c|}
\hline
                                           \multicolumn{3}{|c|}{\multirow{2}{*}{\textbf{\makecell{Design (\#Gates)     (\#Versions)\\94 versions in table, 15 in caption$^\dagger$}}}} & \textbf{\aqed FC} & \multicolumn{3}{c|}{\textbf{\aqedd dFC: Intra-batch FC}} & \multicolumn{3}{c|}{\textbf{\aqedd dFC:  Strong FC}} \\ \cline{4-10}  
                                           & & & \textbf{\begin{tabular}[c]{@{}c@{}}Avg. RT (min)\end{tabular}} & \textbf{\begin{tabular}[c]{@{}c@{}}Avg. RT (min)\end{tabular}} & \textbf{\begin{tabular}[c]{@{}c@{}}\#Bugs\end{tabular}} & \multicolumn{1}{c|}{\textbf{\begin{tabular}[c]{@{}c@{}}\#Sub-Acc.(T/P/C/B)\end{tabular}}} & \textbf{\begin{tabular}[c]{@{}c@{}}Avg. Runtime (min)\end{tabular}} & \textbf{\begin{tabular}[c]{@{}c@{}}\#Bugs\end{tabular}} & \multicolumn{1}{c|}{\textbf{\begin{tabular}[c]{@{}c@{}}\#Sub-Acc.(T/P/C/B)\end{tabular}}}\\\hline
\textbf{AES}~\cite{cong:bandwidth:2017} & (382k)  & (4)  & \makecell[c]{OOM} &   0.97 & 4 & 8 / 7 / 7 / 4 & timeout & 0 & 8 / 7 / 2 / 0  \\\hline
\textbf{ISmartDNN}~\cite{zhang2019skynet} & (42M)  & (3)  & timeout & 0.10 & 2 & 38 / 5 / 5 / 2  & 0.18 & 2 & 38 / 5 / 2 / 2    \\\hline
\textbf{grayscale128}~\cite{piccolboni2019kairos} & (351k) & (5) & timeout &  0.03 & 3 & 3 / 3 / 2 / 2 & 0.07 & 3 & 3 / 3 / 2 / 2    \\\hline
\textbf{grayscale64}~\cite{piccolboni2019kairos} & (194k)  & (5) & timeout &     0.02 & 3 & 3 / 3 / 2 / 2 & 0.02 & 3 & 3 / 3 / 2 / 2 \\\hline
\textbf{grayscale32}~\cite{piccolboni2019kairos} & (106k)  & (5) & 8.20 &  \!\!\!<0.01 & 5 & 3 / 3 / 3 / 3  & 0.30 & 5 & 3 / 3 / 3 / 3 \\\hline
\textbf{mean128}~\cite{piccolboni2019kairos} & (202k)  & (5) & timeout &    0.35 & 3 & 3 / 3 / 2 / 2  & 0.17 & 3 & 3 / 3 / 2 / 2   \\\hline
\textbf{mean64}~\cite{piccolboni2019kairos} & (104k)  & (5) & timeout &    0.38 & 3 & 3 / 3 / 2 / 2  & 0.13 & 3 & 3 / 3 / 2 / 2   \\\hline
\textbf{mean32}~\cite{piccolboni2019kairos} & (54k)  & (5) & 5.53 &  0.17 & 5 & 3 / 3 / 3 / 3  & 0.33 & 5 & 3 / 3 / 3 / 3 \\\hline
\textbf{dnn}~\cite{giordano2021CHIMERA} & (2M)  & (11)  & timeout  &  0.03 & 5 & 34 / 14 / 14 / 5  & 0.13 & 5 & 34 / 14 / 8 / 5   \\\hline
\textbf{nv\_large}~\cite{nvdla} & (16M)   & (23) & timeout &  1.17 & 11 & 89 / 46 / 46 / 11 & 2.93 & 9 & 89 / 46 / 21 / 9  \\\hline
\textbf{nv\_small}~\cite{nvdla} & (1M)  & (23) & timeout  &  0.07 & 11 & 89 / 46 / 46 / 11 & 1.03 & 11 & 89 / 46 / 26 / 11  \\\hline
\end{tabular}
\caption{
\small \textbf{Avg}. \textbf{R}un\textbf{T}imes of FC checks for \aqed and \aqedd. For \aqedd, sub-accelerator counts are provided, including the \textbf{T}otal count that resulted from dFC decomposition, the count with batch sizes greater than one (i.e., \textbf{P}arallel), the count (with batch sizes greater than one) for which FC checks were successful on 1 and 2 batches for intra-batch FC and strong FC respectively, and the count for which \textbf{B}ugs were detected by FC checks. For \aqed FC, experiments could not complete FC check for a single batch in 12 hours (\textbf{timeout}) or exhibited out-of-memory (\textbf{OOM}) errors before timeout. \textbf{Average} runtimes result from dividing the time to detect all bugs by the number of bugs. $^\dagger$\textbf{keypair}~\cite{pqc2019}, \textbf{gsm}~\cite{hara2008chstone}, \textbf{HLSCNN}~\cite{whatmough201916nm}, \textbf{FlexNLP}~\cite{Tambe_isscc2021}, \textbf{Dataflow}~\cite{chi2019rapid}, and \textbf{Opticalflow}~\cite{zhou2018rosetta} all time out for A-QED FC and do not contain any sub-accelerators with batch size greater than one. One OOB bug was detected in \textbf{gsm} and one initialization bug in \textbf{keypair}.}
\label{table:summary_results:dfc}
\end{table*}

\begin{table}
\centering

\setlength{\tabcolsep}{1pt}
\begin{tabular}{|l c c|c|c|c|c|}
\hline
\multicolumn{3}{|c|}{\multirow{4}{*}{\textbf{\textbf{\makecell{Design (\#Gates) (\#Versions)\\Total Versions = 109}}}}}                                                                                                                                     &  \textbf{\aqed RB}                                                              & \multicolumn{3}{c|}{\textbf{\aqedd dRB}}                                                    \\ \cline{4-7}
&                           &                                                    & \textbf{\begin{tabular}[c]{@{}c@{}}Avg. RT\\ (min)\end{tabular}} & \textbf{\begin{tabular}[c]{@{}c@{}}Avg. RT\\ (min)\end{tabular}} & \textbf{\begin{tabular}[c]{@{}c@{}}\#Bugs\end{tabular}} & \multicolumn{1}{c|}{\textbf{\begin{tabular}[c]{@{}c@{}}\#Sub-Acc.\\(T/C/B)\end{tabular}}} \\ \hline
\textbf{AES}~\cite{cong:bandwidth:2017} & (382k)  & (4) & timeout & \multicolumn{2}{c|}{\multirow{11}{*}{\makecell{No RB\\bug detected\\up to input\\sequence\\length\\between\\11 and 24\\depending on\\the design}}}  & 13 / 13 / 0 \\\cline{1-4} \cline{7-7}
\textbf{ISmartDNN}~\cite{zhang2019skynet} & (42M)  & (3)  & timeout  & 
\multicolumn{2}{c|}{} & 32 / 32 / 0     \\\cline{1-4} \cline{7-7}
\textbf{grayscale128}~\cite{piccolboni2019kairos} & (351k) & (5)  & timeout & \multicolumn{2}{c|}{}  &  5 / 5 / 0     \\\cline{1-4} \cline{7-7}
\textbf{grayscale64}~\cite{piccolboni2019kairos} & (194k)  & (5) & timeout & \multicolumn{2}{c|}{}  & 5 / 5 / 0   \\\cline{1-4} \cline{7-7}
\textbf{grayscale32}~\cite{piccolboni2019kairos} & (106k)  & (5) & \multicolumn{1}{c}{} & \multicolumn{2}{c|}{} & 3 / 3 / 0   \\\cline{1-4} \cline{7-7}
\textbf{mean128}~\cite{piccolboni2019kairos} & (202k)  & (5)  & timeout & \multicolumn{2}{c|}{}  & 5 / 5 / 0    \\\cline{1-4} \cline{7-7}
\textbf{mean64}~\cite{piccolboni2019kairos} & (104k)  & (5) & timeout & \multicolumn{2}{c|}{}  & 3 / 3 / 0   \\\cline{1-4} \cline{7-7}
\textbf{mean32}~\cite{piccolboni2019kairos} & (54k)  & (5) & \multicolumn{1}{c}{} & \multicolumn{2}{c|}{}  & 1 / 1 / 0   \\\cline{1-4} \cline{7-7}
\textbf{dnn}~\cite{giordano2021CHIMERA} & (2M)  & (11)  & timeout &\multicolumn{2}{c|}{}  & 5 / 5 / 0     \\\cline{1-4} \cline{7-7}
\textbf{keypair}~\cite{pqc2019} & (>200M) & (1) & timeout & \multicolumn{2}{c|}{} & 21 / 21 / 0    \\\cline{1-4} \cline{7-7}
\textbf{gsm}~\cite{hara2008chstone}& (8.8k) & (1) & timeout & \multicolumn{2}{c|}{}  & 7 / 7 / 0    \\ \hline
\textbf{nv\_large}~\cite{nvdla} & (16M)   & (23) & timeout & \multicolumn{3}{c|}{\multirow{2}{*}{\makecell{ No RB bugs expected}}}  \\\cline{1-4}
\textbf{nv\_small}~\cite{nvdla} & (1M)  & (23) & timeout &  \multicolumn{3}{c|}{}\\\hline
\textbf{HLSCNN}~\cite{whatmough201916nm} & (323k) &  (2)  & timeout & 2.33 & 1 & 25 / 25 / 1    \\\hline
\textbf{FlexNLP}~\cite{Tambe_isscc2021} & (567k) & (9) & timeout & \!\!\!10.77 & 9 & 15 / 15 / 9  \\\hline
\textbf{Dataflow}~\cite{chi2019rapid} & (296k) & (1) & 0.45 & 0.25 & 1 & 9 / 9 / 1   \\\hline
\textbf{Opticalflow}~\cite{zhou2018rosetta} & (555k) & (1) & timeout & 0.17 & 1 & 3 / 3 / 1 \\\hline  \end{tabular}

\caption{
\small  RB checks for \aqed and \aqedd.
For \aqedd, sub-accelerator counts produced by dFC are provided, as in Table~\ref{table:summary_results:dfc}. \aqedd RB checks are performed on all sub-accelerators regardless of batch size, so \textbf{P} is omitted compared to Table~\ref{table:summary_results:dfc}. For \aqed RB, RB checks did not complete even for a input sequence length of 1 within 12 hours (\textbf{timeout}). Sub-accelerators for which RB checks for at least input sequence length of 1 was completed were considered \textbf{C}omplete.  For the first 11 designs, from \textbf{AES} to \textbf{gsm}, no bugs related to unresponsiveness were detected by traditional simulation-based verification. Results are omitted for \textbf{nv\_large} and \textbf{nv\_small}; responsiveness related bugs generally result from parallelism and pipelining, both of which were lost in our manual translation of NVDLA from Verilog to sequential C code. 
}

\label{table:summary_results:drb}
\end{table}

\begin{table}
\centering
\setlength{\tabcolsep}{1pt}
\begin{tabular}{|l c c|c|c|c|c|c|}
\hline
\multicolumn{3}{|c|}{\multirow{4}{*}{\textbf{\textbf{\makecell{Design (\#Gates) (\#Versions)\\Total Versions = 109}}}}} & \multicolumn{4}{c|}{{\textbf{\aqedd dSAC}}} \\  \cline{4-7}  
                                           & &                                                                               & \textbf{\begin{tabular}[c]{@{}c@{}}Avg. RT\\ (min)\end{tabular}} & \textbf{\begin{tabular}[c]{@{}c@{}}\#Bugs\end{tabular}} &  \textbf{\makecell{Bug overlap\\with dFC}} & \multicolumn{1}{c|}{\textbf{\begin{tabular}[c]{@{}c@{}}\#Sub-Acc.\\(T/C/B)\end{tabular}}}  \\ \hline
\textbf{AES}~\cite{cong:bandwidth:2017} & (382k)  & (4)  & 0.12 & 0 & 0 & 8 / 8 / 0 \\\hline
\textbf{ISmartDNN}~\cite{zhang2019skynet} & (42M)  & (3)  & 0.22 & 3 & 2 & 38 / 38 / 3    \\\hline
\textbf{grayscale128}~\cite{piccolboni2019kairos} & (351k) & (5) &  0.04 & 2 & 2 & 3 / 2 / 2     \\\hline
\textbf{grayscale64}~\cite{piccolboni2019kairos} & (194k)  & (5) & 0.01 & 2 & 2 & 3 / 2 / 2  \\\hline
\textbf{grayscale32}~\cite{piccolboni2019kairos} & (106k)  & (5) & \!\!\!<0.01 & 2 & 2 &  3 / 3 / 2  \\\hline
\textbf{mean128}~\cite{piccolboni2019kairos} & (202k)  & (5) & 0.21 & 2 & 2 & 3 / 2 / 2  \\\hline
\textbf{mean64}~\cite{piccolboni2019kairos} & (104k)  & (5) & \!\!\!<0.01 &  2 & 2 & 3 / 2 / 2   \\\hline
\textbf{mean32}~\cite{piccolboni2019kairos} & (54k)  & (5) & \!\!\!<0.01 & 2 & 2 & 3 / 3 / 2  \\\hline
\textbf{dnn}~\cite{giordano2021CHIMERA} & (2M)  & (11)  & 0.01 & 6 & 0 & 34 / 14 / 6  \\\hline
\textbf{keypair}~\cite{pqc2019} & (>200M) & (1) & timeout & 0 & 0 &  14 / 14 / 0  \\\hline
\textbf{gsm}~\cite{hara2008chstone} & (8.8k) & (1) & timeout & 0 &  0 & 5 / 5 / 0  \\\hline
\textbf{nv\_large}~\cite{nvdla} & (16M)   & (23) & 0.84 & 12 & 6 &  89 / 89 / 12   \\\hline
\textbf{nv\_small}~\cite{nvdla} & (1M)  & (23) &  0.11 & 12 & 6 & 89 / 50 / 12  \\\hline
\textbf{HLSCNN}~\cite{whatmough201916nm} & (323k) &  (2)  & 0.45 & 1 & 0 & 25 / 11 / 1   \\\hline
\textbf{FlexNLP}~\cite{Tambe_isscc2021} & (567k) & (9) & timeout & 0 & 0 & 21 / 21 / 0  \\\hline
\textbf{Dataflow}~\cite{chi2019rapid} & (296k) & (1)  & timeout & 0 & 0  & 8 / 8 / 0  \\\hline
\textbf{Opticalflow}~\cite{zhou2018rosetta} & (555k) & (1) & timeout & 0 & 0 & 14 / 14 / 0 \\\hline
\end{tabular}
\caption{
\small SAC checks for \aqedd. Sub-accelerator counts produced by dSAC are provided, as in Table~\ref{table:summary_results:dfc}. \aqedd SAC checks were performed on all sub-accelerators regardless of batch size, so \textbf{P} is omitted compared to Table~\ref{table:summary_results:dfc}. 
}

\label{table:summary_results:dsac}
\end{table}

We demonstrate the practicality and effectiveness of \aqedd for 109 (buggy) versions of several non-interfering LCAs,\footnote{See the appendix for design details and the software artifact~\cite{aqed-decomp-artifact}.} including open-source industrial designs~\cite{nvdla}. We selected these designs for the following reasons:
\begin{itemize}
\item	They 
cover a wide variety of HAs (neural nets, image processing, natural language processing, security). Most are too large for existing off-the-shelf formal tools.
\item   They have been thoroughly verified (painstakingly) using state-of-the-art simulation-based verification techniques. Thus, we can quantify the thoroughness of \aqedd.
\item	With access to buggy versions, we did not have to artificially inject bugs. Bugs we encountered include incorrect initialization, incorrect memory accesses, incorrect array indexing, and unresponsiveness in HLS-generated designs. 

\end{itemize}

Many of the designs were already available in sequential C or C++. We converted Verilog and SystemC designs into sequential C. To facilitate dFC, we manually inserted annotations (like those in Listing~\ref{ex:acc:1}). For \aqed FC, we used CBMC for all designs originally represented in sequential C or C++. For designs in Verilog and SystemC, we used Cadence JasperGold (SystemC designs converted to Verilog via HLS). For \aqedd FC and SAC checks, we used CBMC version 5.10~\cite{kroening2014cbmc}. For \aqed and \aqedd RB checks, we used Cadence JasperGold version 2016.09p002 on Verilog designs generated by the HLS tools used by the designers. Lastly, we used
Frama-C~\cite{framac} to check for initialization and out-of-bounds bugs on the entire C/C++ designs. We ran all our experiments on Intel Xeon E5-2640 v3 with 128GBytes of DRAM.

Tables~\ref{table:summary_results:dfc},~\ref{table:summary_results:drb},~and~\ref{table:summary_results:dsac} summarize our results. We present comparisons between \aqedd (dFC, dRB, dSAC) and \aqed (FC, RB, SAC). Table~\ref{table:summary_results:dfc} also compares \aqedd intra-batch FC vs. \aqedd strong FC 
(cf.~details in the appendix).

    \textbf{Observation 1}: HAs from various domains (including industry) show that non-interfering LCAs are highly common.
    
    \textbf{Observation 2}: The vast majority of the studied HAs are too big for existing off-the-shelf formal verification tools, for both A-QED and conventional formal property verification. 
    
    \textbf{Observation 3}: Table~\ref{table:summary_results:dfc} shows that \aqedd intra-batch FC checks detected bugs inside sub-accelerators (with batch sizes > 1) very quickly---under a minute for almost all of the designs, and just over a minute for nv\_large. For most batch-mode sub-accelerators---except two for each of the following four designs (amounting to eight sub-accelerators in total): grayscale64, grayscale32, mean128, and mean32---intra-batch dFC checks were easily completed using off-the-shelf formal tools. Strong FC checks incur more complexity. Hence, the formal tool timed out after 12 hours for 62 sub-accelerators when running strong FC checks, distributed across multiple designs. Empirically, we found that intra-batch FC checks detected all bugs that were detected by strong FC checks.
    
    \textbf{Observation 4}: 
    \aqedd RB and \aqedd SAC are also highly effective in detecting bugs inside sub-accelerators. For the first 11 designs (\textbf{AES} to \textbf{gsm}) in Table~\ref{table:summary_results:drb}, we do not expect unresponsiveness bugs (confirmed by simulations). Hence, \aqedd RB checks ran for 12 hours (for increasingly longer input sequences) without detecting unresponsiveness.  For designs with RB bugs, \aqedd RB checks on sub-accelerators were able to detect those in less than 11 minutes on average. For \aqedd dSAC, we observed that a significant fraction (26 out of 46 bugs (56\%)) of these bugs were also detected by \aqedd FC checks. Thus, FC alone is effective at catching a wide variety of bugs.
    
    \textbf{Observation 5}: \aqedd detected all bugs that were detected by conventional (simulation-based) verification techniques. Further, all counterexamples produced from verifying sub-accelerators corresponded to real accelerator-level bugs. Compared with traditional simulation-based verification, we report a $\sim5X$ improvement in verification effort on the average, with a $\sim9X$ improvement for the large, industrial NVDLA designs. The overhead of inserting our annotations for dFC can be small compared to what designers already insert to optimize the design. For ISmartDNN, for example, the total number of annotations is 304, which is $2.8\%$ of the total lines of code of the design. In the code of the HLS designs we considered, pragmas amount to~11\% on average.  
    We also observe a $\sim60X$ improvement in average verification runtime compared to conventional simulations.\footnote{The conventional verification effort for NVDLA was based on start and end commit dates in its nv\_small Github repository. The conventional verification runtime for NVDLA, ISmartDNN, and dnn HAs were obtained by running the available simulation tests on our platform. The remaining runtime and effort information were provided by the designers.}


\section{Conclusion}
\label{sec:conclusion}

Our theoretical and experimental results demonstrate that \aqedd is an effective and practical approach for verification of large non-interfering LCAs. \aqedd exploits \aqed principles to decompose a given HA design into sub-accelerators such that \aqed can be naturally applied to the sub-accelerators. \aqedd is especially attractive for HLS-based HA design flows. \aqedd creates several promising research directions: 
\begin{itemize}

\item Extension of our \aqedd experiments to include interfering LCAs (already covered by our theoretical results).

\item Automation of dFC annotations via HLS techniques.

\item dFC approaches beyond our current implementation.

\item Further \aqedd scalability using abstraction.

\item Extension of \aqedd beyond sequential (C/C++) code to include concurrent programs. 

\item Effectiveness of \aqedd for RTL designs (without converting them to sequential C/C++).

\item  Applicability of \aqedd beyond functional bugs (e.g., to detect security vulnerabilities in HAs).

\item Comparison of \aqedd and conventional decomposition.

\item Identifying conditions under which \aqedd is sound.

\end{itemize}


%

%

\section*{Acknowledgment}

This work was supported by the DARPA POSH program (grant FA8650-18-2-7854), NSF (grant A\#:1764000), and the Stanford SystemX Alliance. We thank Prof. David Brooks, Thierry Tambe and Prof. Gu-Yeon Wei from Harvard University, and Kartik Prabhu and Prof. Priyanka Raina from Stanford University for their design contributions in our experiments.

\ifCLASSOPTIONcaptionsoff
  \newpage
\fi

\IEEEtriggeratref{57}


\newpage 

\begin{appendices}

  
  \section{Proofs}

\setcounter{lemma}{0}
\begin{lemma}[Soundness of FC]
If an HA is functionally correct, then it is functionally consistent.
\end{lemma}

\begin{proof}
Let $\acc$ be an HA that is functionally correct, and let
$\inpseq$ be a sequence of input batches.  Assume
$\inpseq = \langle \inp_1, \ldots, \inp_n \rangle$, $\stateseq = \generatedstateseq(\inpseq, \state_0)$, $\initstateseq = \getseqinitstates(\stateseq) = \langle \initstateargs{1}, \ldots, \initstateargs{n} \rangle$, and
$\outpseq = \getoutputstate(\getseqfinalstates(\stateseq)) = \langle \outp_1, \ldots, \outp_n \rangle$.  Further, let $i\in[1,n]$ and $j,j'\in[1,b]$, and suppose $\inp_i(j)=\inp_n(j')$ and $\getrelstate(\initstateargs{i}) = \getrelstate(\initstateargs{n})$.
By functional correctness we have
\[\outp_i(j) = \specout(\inp_i(j),
\getrelstate(\initstateargs{i})).\]
Similarly, we also have
\[\outp_n(j') = \specout(\inp_n(j'),
\getrelstate(\initstateargs{n})).\]
But $\inp_i(j) = \inp_{n}(j')$ and
$\getrelstate(\initstateargs{i}) = \getrelstate(\initstateargs{n})$,
so it follows that $\outp_i(j)=\outp_n(j')$.
\end{proof}

\begin{lemma}[Soundness of Strong FC]
If an HA is functionally correct, then it is strongly functionally consistent.
\end{lemma}

\begin{proof}
Let $\acc$ be an HA that is functionally correct, and let
$\state_0$ and $\state_0'$ be two arbitrary but fixed reachable initial states,
such that $\getrelstate(\state_0) = \getrelstate(\state_0')$. By
Definition~\ref{def:reachable} of reachable states, there exist
initial states $\initstate, \initstate' \in \initstates$ and sequences
$\inpseq$ and $\inpseq'$ of input batches such that $\state_0 \in
\generatedstateseq(\inpseq, \initstate)$ and $\state_0' \in
\generatedstateseq(\inpseq', \initstate')$. Without loss of
generality, we assume that $\state_0$ and $\state_0'$ are the last
initial states that appear in $\generatedstateseq(\inpseq, \initstate)$
and $\generatedstateseq(\inpseq', \initstate')$, respectively.

Let $\inp$ and $\inp'$ be arbitrary but fixed input batches such that
$\inp(j) = \inp'(j')$ for arbitrary but fixed $j$ and $j'$.

Let $\overline{\inpseq}$ be the result of replacing the last element of $\inpseq$ with $\inp$,
and let $\overline{\inpseq'}$ be the result of replacing the last element of $\inpseq'$ with $\inp'$,
We now apply functional correctness with
respect to the two sequences $\overline{\inpseq}$ and $\overline{\inpseq'}$ of input batches and the initial states $\initstate$
and $\initstate'$.

Let $\finalstate$ and $\finalstate'$ be the final states reached by
executing $\overline{\inpseq}$ from $\initstate$ and by executing
$\overline{\inpseq'}$ from $\initstate'$.
Let $\outp$ and $\outp'$ be the output batches produced in the final
states $\finalstate$ and $\finalstate'$, respectively.

Since $\acc$ is functionally correct, for $\overline{\inpseq}$ we
have $\outp(j) = \specout(\inp(j),
\getrelstate(\state_0))$. Further, for $\overline{\inpseq'}$, we have $\outp'(j') = \specout(\inp'(j'),
\getrelstate(\state_0'))$.

Since we have $\inp(j) = \inp'(j')$ and
$\getrelstate(\state_0) = \getrelstate(\state_0')$,
we must have $\outp(j) = \outp'(j')$. Therefore $\acc$ is strongly functionally
consistent.
\end{proof}

\begin{lemma}[Completeness of SAC + FC + SC]
If an HA is strongly connected and single-action correct and has a bug, then it is
not functionally consistent.
\end{lemma}

\begin{proof}
The idea of the proof is similar to the proof of completeness by
Proposition~1 in~\cite{singh:aqed:2020}.

Let $\acc$ be an HA that is strongly connected, single-action
correct, and that has a bug.  That is, by definition of a bug (i.e.,
negation of functional correctness by
Definition~\ref{def:corr:NEW:MODEL}), there exists an initial state
$\state_0 \in \initstates$ and a sequence $\inpseq$ of input batches
such that $\outp_n(j_{\mathit{bug}}) \not =
\specout(\inp_n(j_{\mathit{bug}}), \getrelstate(\initstateargs{n}))$
for some arbitrary but fixed $1 \leq j_{\mathit{bug}} \leq
\batchsize$, where $\inpseq = \langle \inp_1, \ldots, \inp_n \rangle$,
$\stateseq = \generatedstateseq(\inpseq, \state_0)$, $\initstateseq =
\getseqinitstates(\stateseq) = \langle \initstateargs{1}, \ldots,
\initstateargs{n} \rangle$, $\finalstateseq =
\getseqfinalstates(\stateseq) = \langle \finalstateargs{1}, \ldots,
\finalstateargs{n} \rangle$, and $\outpseq =
\getoutputstate(\finalstateseq) = \langle \outp_1, \ldots, \outp_n
\rangle$.

Note that $\initstateargs{n}$ is the initial state in which,
given input batch $\inp_n$, the computation started that caused the incorrect
output batch $\outp_n$ to be produced in the final state $\finalstateargs{n}$.

Since $\initstateargs{n} \in \generatedstateseq(\inpseq, \state_0)$,
by Definition~\ref{def:reachable} $\initstateargs{n}$ is reachable,
and hence also its relevant state $\getrelstate(\initstateargs{n})$ is
reachable.

Since $\getrelstate(\initstateargs{n})$ is reachable, by single-action
correctness there exists some reachable initial state $\state \in \states$
with $\getrelstate(\initstateargs{n}) =
\getrelstate(\state)$ such that executing input batch $\inp_n$ in
$\state$ produces a correct result for the batch element
$\inp_n(j_{\mathit{bug}})$ in the respective final state $\finalstate'$
of that execution.

By strong connectedness, since $\finalstateargs{n}$ and $\state$ are
both reachable, there exists a sequence $\inpseq'$ of input batches
that transitions $\acc$ from $\finalstateargs{n}$ to $\state$.

In state $\state$, we execute input batch $\inp_n$. By single-action
correctness, the result of executing $\inp_n$ produces a correct
output batch $\outp$ such that $$\outp(j_{\mathit{bug}}) =
\specout(\inp_n(j_{\mathit{bug}}), \getrelstate(\state)).$$ Since
$\getrelstate(\initstateargs{n}) = \getrelstate(\state)$, we have
$$\specout(\inp_n(j_{\mathit{bug}}), \getrelstate(\state)) =
\specout(\inp_n(j_{\mathit{bug}}), \getrelstate(\initstateargs{n})).$$

Finally, since $$\outp_n(j_{\mathit{bug}}) \not =
\specout(\inp_n(j_{\mathit{bug}}), \getrelstate(\initstateargs{n}))$$
and $$\outp(j_{\mathit{bug}}) = \specout(\inp_n(j_{\mathit{bug}}),
\getrelstate(\state)),$$ also $$\outp_n(j_{\mathit{bug}}) \not =
\outp(j_{\mathit{bug}}),$$ hence $\acc$ is not functionally consistent.  
\end{proof}

\begin{lemma}[Completeness of SAC + Strong FC]
If an HA is single-action correct and has a bug, then it is
not strongly functionally consistent.
\end{lemma}

\begin{proof}  
  Let $\acc$ be an HA that is single-action
correct and that has a bug. That is, by definition of a bug (i.e., negation of functional correctness by Definition~\ref{def:corr:NEW:MODEL}), there exists 
an initial state $\state_I \in
\initstates$ and a sequence $\inpseq$ of input batches 
such that $\outp_n(j_{\mathit{bug}}) \not = \specout(\inp_n(j_{\mathit{bug}}),
\getrelstate(\initstateargs{n}))$ for some arbitrary but fixed $1 \leq j_{\mathit{bug}} \leq \batchsize$, 
where $\inpseq = \langle \inp_1, \ldots, \inp_n \rangle$, $\stateseq =
\generatedstateseq(\inpseq, \state_I)$, $\initstateseq =
\getseqinitstates(\stateseq) = \langle \initstateargs{1}, \ldots,
\initstateargs{n} \rangle$, $\finalstateseq =
\getseqfinalstates(\stateseq) = \langle \finalstateargs{1}, \ldots,
\finalstateargs{n} \rangle$, and $\outpseq =
\getoutputstate(\finalstateseq) = \langle \outp_1, \ldots, \outp_n
\rangle$.

We must show that $\acc$ is not strongly FC, i.e., using notation from
Definition~\ref{def:strong:fc}, for some $j, j', \inp, \inp',
\state_0$ and $\state_0'$ with $\inp(j) = \inp'(j')$ and
$\getrelstate(\state_0) = \getrelstate(\state_0')$, we have $\outp(j)
\not = \outp'(j')$. Without loss of generality, assume $j =
j' = j_{\mathit{bug}}$. 

Note that $\initstateargs{n}$ is the initial state in which,
given input batch $\inp_n$, the computation started that caused the incorrect
output batch $\outp_n$ to be produced in the final state $\finalstateargs{n}$.
  
Since $\initstateargs{n} \in \generatedstateseq(\inpseq, \state_I)$,
by Definition~\ref{def:reachable} $\initstateargs{n}$ is reachable,
and hence also its relevant state $\getrelstate(\initstateargs{n})$ is
reachable.

Since $\initstateargs{n}$ is a reachable initial state,  
let $\state_0 = \initstateargs{n}$, which is possible by Definition~\ref{def:strong:fc}, and let $\inp = \inp_n$, i.e.,
$\state_0$ is the state for which $\acc$ produces an incorrect output
for input batch $\inp$ with batch element $\inp(j_\mathit{bug})$.

Further, let $\stateseq =
\generatedstateseq(\langle \inp \rangle, \state_0)$, $\finalstateseq =
\getseqfinalstates(\stateseq) = \langle \finalstate \rangle$, and
$\outpseq = \getoutputstate(\finalstateseq) = \langle \outp \rangle$.
Due to the bug, we have $\outp(j) \not = \specout(\inp(j),
\getrelstate(\state_0))$. 

Since $\getrelstate(\initstateargs{n})$ is reachable, by single-action
correctness (SAC) there exists some reachable initial state $\state
\in \states$ with $\getrelstate(\initstateargs{n}) =
\getrelstate(\state)$ such that executing input batch $\inp_n$ in
$\state$ produces a correct result for the batch element
$\inp_n(j_{\mathit{bug}})$ in the respective final state
$\finalstate'$ of that execution (without loss of generality, we
assume that that batch element occurs at index $j_{\mathit{bug}}$ by
SAC).

Towards showing inconsistency, let $\inp' = \inp_n$ and $\state_0' =
\state$ be a reachable initial state, where 
state $\state$ exists by SAC, such that $\getrelstate(\state_0')
= \getrelstate(\state) = \getrelstate(\initstateargs{n})$. 
Further, let $\stateseq' = \generatedstateseq(\langle
\inp' \rangle, \state_0')$, $\finalstateseq' =
\getseqfinalstates(\stateseq') = \langle \finalstate' \rangle$, and
$\outpseq' = \getoutputstate(\finalstateseq') = \langle \outp'
\rangle$.

By single-action correctness and since $j =
j' = j_{\mathit{bug}}$, we have $$\outp'(j') = \specout(\inp'(j'),
\getrelstate(\state_0')).$$ 

Since $\state_0 = \initstateargs{n}$, also $\getrelstate(\state_0) =
\getrelstate(\initstateargs{n})$. Further, since
$\getrelstate(\state_0') = \getrelstate(\state) =
\getrelstate(\initstateargs{n})$, also $$\getrelstate(\state_0) =
\getrelstate(\state_0')$$ by transitivity.

Since   $\getrelstate(\state_0) = \getrelstate(\state_0')$ and $\inp(j)
= \inp'(j')$, also $$\specout(\inp(j), \getrelstate(\state_0)) =
\specout(\inp'(j'), \getrelstate(\state_0')).$$

Further, since $$\outp(j) \not = \specout(\inp(j),
\getrelstate(\state_0))$$ due to the bug and also $$\specout(\inp(j),
\getrelstate(\state_0)) = \specout(\inp'(j'),
\getrelstate(\state_0')),$$
and $$\outp'(j') = \specout(\inp'(j'),
\getrelstate(\state_0')),$$ 
we have $$\outp(j) \not = \outp'(j')$$ and hence
$\acc$ is not strongly functionally consistent.
\end{proof}

\begin{lemma}[Functional Composition and Strong FCD]
Let $\acc_0 = \acc_2 \circ \acc_1$. 
If both $\acc_1$
and $\acc_2$ are strongly FCD then $\acc_0$ is
strongly FCD.
\end{lemma}

\begin{proof}[Proof Sketch]

  Assume that $\acc_1$ and $\acc_2$ are strongly FCD.
  
  \begin{enumerate}

  \item To show that $\acc_0$ is strongly FCD, assume that the
    antecedent in the implication of Definition~\ref{def:strong:fcd}
    holds for $\acc_0$, i.e., we have $\inp_0(j_0) = \inp_0'(j_0')
    \wedge \getrelstate(\state_0) = \getrelstate(\state_0')$, where
    $\inp_0, \inp_0', j_0, j_0', \state_0$, and $\state_0'$ are
    arbitrary but fixed. We must show that the conclusion holds, i.e.,
    $\outp_0(j_0) = \outp_0'(j_0') \wedge
    \getrelstate(\finalstateargs{0}) =
    \getrelstate(\finalstateargs{0}')$, where $\outp_0, \outp_0'$ and
    $\finalstateargs{0}, \finalstateargs{0}'$ are the output batches
    and final states generated by $\acc_0$ following
    Definition~\ref{def:strong:fcd}. \label{proof_step_1}
    
  \item Given $\inp_0(j_0) = \inp_0'(j_0') \wedge
    \getrelstate(\state_0) = \getrelstate(\state_0')$,
    and since $\acc_0 = \acc_2 \circ \acc_1$,
    let $\inp_1, \inp_1', j_1, j_1', \state_1, \state_1'$ of $\acc_1$
    such that $\inp_1 = \inp_0$, $j_1 = j_0$, $\inp_1' = \inp_0'$,
    $j_1' = j_0'$, $\getrelstate(\state_1) = \getrelstate(\state_0)$,
    $\getrelstate(\state_1') = \getrelstate(\state_0')$. By
    transitivity and since $\getrelstate(\state_0) =
    \getrelstate(\state_0')$, also $\getrelstate(\state_1) =
    \getrelstate(\state_1')$. \label{proof_step_2}
    
  \item Due to $\inp_0(j_0) = \inp_0'(j_0') \wedge
    \getrelstate(\state_0) = \getrelstate(\state_0')$ and
    step~\ref{proof_step_2}, we have $\inp_1(j_1) = \inp_1'(j_1')
    \wedge \getrelstate(\state_1) = \getrelstate(\state_1')$, i.e.,
    the antecedent of the implication of
    Definition~\ref{def:strong:fcd} holds for $\acc_1$.
    \label{proof_step_3}
    
  \item Due to step~\ref{proof_step_3} and since $\acc_1$ is strongly FCD, we conclude $\outp_1(j_1) = \outp_1'(j_1') \wedge \getrelstate(\finalstateargs{1}) = \getrelstate(\finalstateargs{1}')$, where $\outp_1, \outp_1'$ and
    $\finalstateargs{1}, \finalstateargs{1}'$ are the output batches
    and final states generated by $\acc_1$ following
    Definition~\ref{def:strong:fcd}.
    \label{proof_step_4}
    
  \item 
    Let $\state_2$ and $\state_2'$ be initial states of $\acc_2$ 
    obtained from $\finalstateargs{1}$ and $\finalstateargs{1}'$, respectively, such
    that $\getrelstate(\state_2) = \getrelstate(\finalstateargs{1})$
    and $\getrelstate(\state_2') =
    \getrelstate(\finalstateargs{1}')$. By transitivity and
    step~\ref{proof_step_4}, also $\getrelstate(\state_2) =
    \getrelstate(\state_2')$.
\label{proof_step_5}

\item Let $\inp_2, \inp_2'$, $j_2, j_2'$ such that 
   $\inp_2 = \outp_1$, $\inp_2' = \outp_1'$, $j_2
  = j_1$, and $j_2' = j_1'$.\label{proof_step_6}
  
\item Due to steps~\ref{proof_step_4} and~\ref{proof_step_6}, we have
  $\inp_2(j_2) = \inp_2'(j_2') \wedge \getrelstate(\state_2) =
  \getrelstate(\state_2')$, i.e., the antecedent of the implication of
  Definition~\ref{def:strong:fcd} holds for $\acc_2$.
  \label{proof_step_7}
  
\item Due to step~\ref{proof_step_7} and since $\acc_2$ is strongly
  FCD, we conclude $\outp_2(j_2) = \outp_2'(j_2') \wedge
  \getrelstate(\finalstateargs{2}) =
  \getrelstate(\finalstateargs{2}')$, where $\outp_2, \outp_2'$ and
  $\finalstateargs{2}, \finalstateargs{2}'$ are the output batches and
  final states generated by $\acc_2$ following
  Definition~\ref{def:strong:fcd}.  \label{proof_step_8}
  
\item Due to step~\ref{proof_step_8}, Definition~\ref{def:composable} and Definition~\ref{def:acc:decomp2}, and since
  $\outputs_2 = \outputs_0$, we have, for
  output batches $\outp_0, \outp_0'$ of $\acc_0$, $\outp_0 = \outp_2$
  and $\outp_0' = \outp_2'$.

  Further, since $j_0 = j_2 = j_1$ and $j_0' =
  j_2' = j_1'$, also $j_0 = j_2$ and $j_0' = j_2'$, and hence
  $\outp_0(j_0) = \outp_0'(j_0')$. Finally, since
  $\getrelstate(\finalstateargs{2}) =
  \getrelstate(\finalstateargs{2}')$ and also
  $\getrelstate(\finalstateargs{0}) =
  \getrelstate(\finalstateargs{0}')$. Hence $\acc_0$ is strongly FCD. \qedhere
  \label{proof_step_9} 
    
  \end{enumerate}
  
\end{proof}

\setcounter{theorem}{0}
\begin{theorem}[Completeness of \aqedd]
Let $\acc_0, \acc_1$, and $\acc_2$ be HAs such that $\acc_0 =
\acc_2 \circ \acc_1$ and $\acc_0$ is single-action correct.
If $\acc_1$ and $\acc_2$ are strongly FCD then $\acc_0$ is
functionally correct.
\end{theorem}

\begin{proof}
Let $\acc_0, \acc_1$, and $\acc_2$ be HAs such that $\acc_0 =
\acc_2 \circ \acc_1$ and $\acc_0$ is single-action correct.

Assume that $\acc_1$ and $\acc_2$ are strongly FCD.

By Lemma~\ref{prop:fc:decomp}, $\acc_0$ is strongly FCD, and further
by Corollary~\ref{cor:strong:fcd:fc}, $\acc_0$ is strongly FC.

By contraposition of Lemma~\ref{lem:strong:FC:complete:NEW:MODEL}, if
an HA is single-action correct and strongly FC, then it is
functionally correct.

Hence since $\acc_0$ is single-action correct and strongly FC,
$\acc_0$ is also functionally correct.
\end{proof}


\section{Generic \aqed}

\begin{lstlisting}[label={algo:aqed}, language=C, caption={Generic \aqedd Module for dFC}, breaklines=true,basicstyle=\scriptsize,     keywordstyle=\color{maroon}\bfseries, commentstyle=\color{darkgreen}\textit, stringstyle=\color{darkpink}\ttfamily, numbers=left, escapeinside={(*}{*)}, numbers=left, xleftmargin=2em,frame=single,framexleftmargin=2em, numberstyle=\tiny]
(*\# DEFINE \bfseries BATCH\_MEM\_IN, BATCH\_MEM\_OUT, IN\_SIZE, OUT\_SIZE, NUM\_OP, IN\_ALLOC\_RULE, OUT\_ALLOC\_RULE*) 

orig_val[(*\bfseries IN\_SIZE*)], dup_val[(*\bfseries IN\_SIZE*)], orig_out[(*\bfseries OUT\_SIZE*)], orig_labeled = 0, dup_labeled = 0, in_ct = 0, out_ct = 0, orig_idx = 0, dup_idx = 0, dup_done = 0, fc_check = 0 //Initialisation

(*\textbf{aqed\_in}*)(in, orig, dup)\{ //monitors inputs
    label_orig = orig & !orig_labeled; match = 1;
    for(i = 0; i < (*\textbf{IN\_SIZE}*); i++){
        if(in[i] != orig_val[i]){
            match = 0;  break;
        }
    }
    label_dup = dup & !dup_labeled & orig_labeled & match;
    if(label_orig){
        orig_labeled = 1;
        orig_idx = in_ct;
        for(i = 0; i < (*\textbf{IN\_SIZE}*); i++)
            orig_val[i] = in[i];
    }
    if(label_dup){
        dup_labeled = 1;
        dup_idx = in_ct;
    }
    in_ct++;
}

(*\textbf{aqed\_out}*)(out)//analyzes outputs
    if(orig_labelled & out_ct == orig_idx & !dup_done){
        for(i = 0; i < (*\textbf{OUT\_SIZE}*); i++){
            orig_out[i] = out[i];
        }
    }
    if(orig_labeled & dup_labeled & out_ct==dup_idx & !dup_done){
        dup_done = 1; fc_check = 1;
        for(i = 0; i < (*\textbf{OUT\_SIZE}*); i++){
            if(orig_out[i] == out[i])
                fc_check = 0; break;
        }
    }
    if(out_ct[i] > dup_idx)
        dup_done = 1;

    out_ct++;
    return {dup_done, fc_check};
}

(*\bfseries aqed\_top*)((*\bfseries BATCH\_MEM\_IN*), orig[(*\bfseries NUM\_OP*)], dup[(*\bfseries NUM\_OP*)]){ //Top function
    
    in_seq = (*\bfseries Parallel2Serial1*)((*\bfseries BATCH\_MEM\_IN*), (*\bfseries IN\_SIZE*), (*\bfseries OP\_SIZE*), (*\bfseries IN\_ALLOC\_RULE*)); //Parallel-to-serial converter
    for(i = 0; i < (*\bfseries NUM\_OP*)*(*\bfseries IN\_SIZE*); i = i + (*\bfseries IN\_SIZE*)){
        (*\bfseries aqed\_in*)(in_seq[i], orig[i], dup[i]);
    }

    (*\bfseries acc*)((*\bfseries BATCH\_MEM\_IN*), (*\bfseries BATCH\_MEM\_OUT*)); //Accelerator function

    out_seq = (*\bfseries Parallel2Serial2*)((*\bfseries BATCH\_MEM\_OUT*), (*\bfseries OUT\_SIZE*), (*\bfseries OP\_SIZE*), (*\bfseries OUT\_ALLOC\_RULE*)); //Parallel-to-serial converter
    for(i = 0; i < (*\bfseries NUM\_OP*)*(*\bfseries OUT\_SIZE*); i = i + (*\bfseries OUT\_SIZE*)){
        {dup_done, fc_check} = (*\bfseries aqed\_out*)(out_seq[i]);
    }
    return {dup_done, fc_check} 
}

\end{lstlisting}


\section{Decomposition for RB: \texorpdfstring{\MakeLowercase{d}RB}{dRB}}
\label{app:sec:drb}

Given an HA, responsiveness (RB) checks with respect to some bound $n$
based on Definition~\ref{def:rb} are carried out in \aqedd for
sub-accelerators of HA. These sub-accelerators can be (and often are) different from the
sub-accelerators checked for FC. That is because RB involves a much
simpler check than FC as the only requirement is that \emph{some}
(i.e., not necessarily correct) output is produced within a response
time of $n$ clock cycles. Bound $n$ is a global constant which is used for
RB checks of all sub-accelerators as an upper bound of the
response time. The rationale is that if a sub-accelerator of HA fails
an RB check, then HA also fails an RB check. Hence RB checking is
used for bug-hunting rather than a real-time check for response time
guarantees.

To identify sub-accelerators for RB checking, the C/C++ representation
of the HA is converted to static single assignment form (SSA) and
loops are fully unrolled.  SSA form is not a requirement but makes it
easier to apply the following \emph{sliding window algorithm} to
identify sub-accelerators. The algorithm has two phases, an
\emph{enlarging phase} and a \emph{shrinking phase}.

Sub-accelerators of the given HA are represented by sequences of
contiguous lines of code (LOCs) that are part of the SSA
representation of the accelerator function of HA. A \emph{window $W$
  of code} is used to select such sequences of LOCs. The current size
of $W$ is given by the number of LOCs that fit in $W$ and it changes
dynamically during a run of the algorithm. The inputs and outputs of
the sub-accelerator given by current $W$ are determined as follows.
Variables that are never updated or assigned in
$W$ are considered inputs of the respective sub-accelerator, and
variables in $W$ that are used only to update variables outside of $W$ are considered
outputs. Given $W$, we synthesize the corresponding sub-accelerator
and run RB checks on the resulting RTL representation using symbolic starting states, as certain RB
bugs express only at the RTL level.

Initially, the size of $W$ is set to a small constant that is a
parameter. Further, the first LOC in the initial $W$ is the first LOC
of the accelerator function of HA.

Given this initial configuration, the algorithm enters the enlarging
phase. There, RB checks with respect to bound $n$ are run repeatedly
for the sub-accelerator that corresponds to the LOCs in current
$W$. If an RB check fails, the algorithm terminates and returns the
corresponding counterexample trace. If an RB check passes, then the
size of $W$ is increased by keeping the top boundary of $W$ in place and
moving its bottom boundary downwards by $\delta$ LOCs, which is a
parameter. The next RB check is run with respect to the enlarged $W$.

If during the enlarging phase an RB check times out, where the time
out value is a parameter, then the algorithm enters the shrinking
phase. Timing out indicates that the size of the sub-accelerator in
the current $W$ was prohibitively large and prevented the unrolling of the
sub-accelerator up to $n$ clock cycles in BMC. In the shrinking phase,
RB checks are run repeatedly with respect to the current $W$. If an RB
check times out, then the size of $W$ is reduced by moving
the top boundary of $W$ downwards by $\delta$ LOCs and keeping its bottom boundary in place.
If an RB check fails, then, like in the enlarging phase, the algorithm
terminates and returns the corresponding counterexample trace. If an
RB check passes, then the algorithm enters the enlarging phase again. The
algorithm terminates if $W$ reaches the end of the code.

In practice, due to time outs the sliding window algorithm may not be
able to run passing or failing RB checks for all possible
sub-accelerators of HA. However, as our experimental results show, it
was effective in detecting failing RB checks in sub-accelerators where
RB checks on the full HA timed out.

\section{Results (Extended)}
\label{app:results}
\begin{table}[h]
\centering
\setlength{\tabcolsep}{1.5pt}
\begin{tabular}{|c|c|r|c|r|r|r|} 
\hline
\textbf{Source} & \textbf{Design} & \makecell{\textbf{\#Ver.}} & \textbf{\makecell{Source\\Code}} & \multicolumn{3}{c|}{\textbf{Size\tnote{$\dagger$}}}  \\
 & & & & PLOC & FFs & Gates \\\hline
\makecell{\textbf{AES}\\\textbf{Encryption}\\~\cite{cong:bandwidth:2017}} & \textbf{AES} & 4 & C++ & 220 & 82k & 382k   \\\hline
\textbf{SkyNet}~\cite{zhang2019skynet} & \textbf{ISmartDNN} & 3 & C++ & 1,038 & 9.3M & 41.9M \\\hline
\makecell{\textbf{KAIROS}~\cite{piccolboni2019kairos}} & \textbf{grayscale128} & 5& C & 202 & 28k & 351k  \\\cline{2-7}
 & \textbf{grayscale64} & 5 & C &  202 & 14k & 194k  \\\cline{2-7}
 & \textbf{grayscale32} & 5 & C & 202 & 8k & 106k  \\\cline{2-7}
 & \textbf{mean128} & 5 & C & 262 & 13k & 202k  \\\cline{2-7}
 & \textbf{mean64} & 5 & C & 262 & 6k & 104k  \\\cline{2-7}
 & \textbf{mean32} & 5 & C & 262 & 4k & 54k  \\\hline
\textbf{CHIMERA}~\cite{giordano2021CHIMERA} & \textbf{dnn} & 11 & SystemC & 3,263  & 167k &  2M \\\hline
\textbf{NVDLA}~\cite{nvdla} & \makecell{\textbf{nv\_large}} & 23 & Verilog & 542k &  413k & 16M  \\\cline{2-7}
 & \makecell{\textbf{nv\_small}} & 23 & Verilog & 330k & 47k  & 1M  \\\hline
\textbf{HLSCNN}~\cite{whatmough201916nm} & \textbf{HLSCNN} & 2 & SystemC & 8,280  & 26k & 323k  \\\hline
\textbf{FlexNLP}~\cite{Tambe_isscc2021} & \textbf{FlexNLP} & 9 & SystemC & 6,062 & 30k  & 567k     \\\hline
\makecell{\textbf{Custom}\\\textbf{Design}~\cite{chi2019rapid}} & \textbf{Dataflow}  & 1 & C++ & 104 & 85k & 296k  \\\hline
\textbf{Rosetta}~\cite{zhou2018rosetta} & \textbf{Opticalflow}  & 1 & C++ & 444 & 136k & 555k  \\\hline
\textbf{PQC}~\cite{pqc2019} & \textbf{keypair} & 1 & C & 4,460 & >70M & >200M  \\\hline
\textbf{CHStone}~\cite{hara2008chstone} & \textbf{gsm}  & 1 & C & 395 & 496 & 8.8k  \\\hline
& \multicolumn{2}{c|}{\textbf{Total Versions = 109}} & & & &  \\\hline
\end{tabular}
\caption{Details of designs used in our experiments.}
\label{table:summary_design}
\end{table}

\begin{table*}[t]
\centering
\setlength{\tabcolsep}{1.5pt}
\begin{tabular}{|c|r|r|r|r|r|r|c|r|c|} 
\hline
 \textbf{Design} & \textbf{\makecell{\#Ver.}} &   \multicolumn{2}{c|}{\textbf{\makecell{Conventional\\Simulation}}} &  \multicolumn{5}{c|}{\textbf{\aqedd}} & \\\cline{3-9}         
 &  & \textbf{\makecell{Effort\\(PD)}} & \textbf{\makecell{Runtime\\ (minute)}} & \textbf{\makecell{Effort\\(PD)}} & \multicolumn{4}{c|}{\makecell{\textbf{Runtime (minute)} [min, avg, max] }} & \textbf{Bug Type} \\\cline{6-9}
 & & & \makecell{[min, avg, max]} & & \textbf{\makecell{dFC: intra-batch\\ FC}} & \textbf{\makecell{dFC: strong\\ FC}} &  \textbf{\makecell{dRB}} & \textbf{\makecell{dSAC}} & \\\hline
 \textbf{AES} &  4 &  1 & \makecell[r]{60, 240, 480}  &  2  & \makecell[r]{0.67, 0.97, 1.57}  & timeout & No RB & \makecell[r]{0.08, 0.12, 0.64 } & FC \\\cline{1-7} \cline{9-10}
 \textbf{ISmartDNN} & 3 &  N/A & \makecell[r]{31, 35, 37} & 5  & \makecell[r]{0.05, 0.10, 0.17 } & \makecell[r]{0.13, 0.18, 0.23 } & bugs detected & \makecell[r]{0.03, 0.22, 0.46 } & SAC  \\\cline{1-7} \cline{9-10}
 \textbf{grayscale128} & 5 &  N/A &  N/A & 1 & \makecell[r]{0.01, 0.03, 0.53 } & \makecell[r]{0.04, 0.07, 0.34 } & up to input & \makecell[r]{0.01, 0.04, 0.07 } & FC, SAC \\\cline{1-7} \cline{9-10}
\textbf{grayscale64} & 5 &  N/A &  N/A &  1 & \makecell[r]{0.01, 0.02, 0.04} & \makecell[r]{0.01, 0.02, 0.07 } & sequence & \makecell[r]{<0.01, 0.01, 0.04 } & FC, SAC \\\cline{1-7} \cline{9-10}
\textbf{grayscale32} & 5 &  N/A &  N/A &  1 & \makecell[r]{<0.01, <0.01, 0.01 } & \makecell[r]{<0.01, 0.30, 0.56 } & length  & \makecell[r]{<0.01, <0.01, 0.02 } & FC, SAC \\\cline{1-7} \cline{9-10}
\textbf{mean128} & 5 &  N/A &   N/A &  1 & \makecell[r]{0.02, 0.35, 0.60 } & \makecell[r]{0.10, 0.17, 0.53 } & between & \makecell[r]{0.07, 0.21, 0.35 } & FC, SAC \\\cline{1-7} \cline{9-10}
\textbf{mean64} & 5 & N/A &  N/A &  1 &  \makecell[r]{0.17, 0.38, 0.53 } & \makecell[r]{0.03, 0.13, 0.35 } & 11 \& 24 & \makecell[r]{<0.01, <0.01, 0.07 } & FC, SAC \\\cline{1-7} \cline{9-10}
\textbf{mean32} & 5 &   N/A &  N/A &  1 & \makecell[r]{<0.01, 0.17, 0.57 } & \makecell[r]{<0.01, 0.33, 0.41 } & depending on  & \makecell[r]{<0.01, <0.01, 0.02 } & FC, SAC  \\\cline{1-7} \cline{9-10}
\textbf{dnn} & 11 &  14 & \makecell[r]{20, 25, 30} & 4 &  \makecell[r]{0.02, 0.03, 0.05 } & \makecell[r]{0.08, 0.13, 0.35 } & the design  & \makecell[r]{<0.01, 0.01, 0.06 }    & FC, SAC\\\cline{1-7} \cline{9-10}
\textbf{keypair} & 1 &  14  & \makecell[r]{360, 360, 360}  & 10 &  \multicolumn{2}{c|}{\makecell{Sub-acc of batch size >1 not found.}}  &   & timeout & OOB  \\\cline{1-5} \cline{9-10} 
\textbf{gsm} & 1 &  N/A & \makecell[r]{0.17, 0.17, 0.17} & 2 & \multicolumn{2}{c|}{\makecell{Frama-C~\cite{framac} catches Init. and OOB Bugs}}    &  & timeout & init.  \\\hline
\textbf{nv\_large} & 23 & 270 & \makecell[r]{30, 40, 120} & 15 & \makecell[r]{0.03, 1.17, 5.00} & \makecell[r]{0.53, 2.93, 12.17 } & No bugs & \makecell[r]{0.08, 0.84, 1.34 } & FC, SAC   \\\cline{1-2} \cline{4-7}  \cline{9-10}
\textbf{nv\_small} & 23  & overall & \makecell[r]{30, 40, 120} & 15 & \makecell[r]{0.03, 0.07, 0.15 } & \makecell[r]{0.35, 1.03, 10.18 } & expected & \makecell[r]{0.03, 0.11, 0.45 } & FC, SAC \\\hline
\textbf{HLSCNN} & 2 &  60 & \makecell[r]{3, 3, 3} & 5 & \multicolumn{2}{c|}{Sub-acc. of batch size > 1}  & \makecell[r]{2.33, 2.33, 2.33}  & \makecell[r]{0.45, 0.45, 0.45} & RB, SAC    \\\cline{1-5} \cline{8-10}
\textbf{FlexNLP} & 9  & 60 & \makecell[r]{3, 3, 3} & 5 & \multicolumn{2}{c|}{} & \makecell[r]{4.11, 10.77, 20.20 }  & timeout & RB  \\\cline{1-5} \cline{8-10}
\textbf{Dataflow} & 1 &  N/A & N/A & 2 & \multicolumn{2}{c|}{not found}  & \makecell[r]{0.25, 0.25, 0.25}  & timeout & RB  \\\cline{1-5} \cline{8-10}
\textbf{Opticalflow} &  1 & 7 & \makecell[r]{10, 10, 10} & 6 &  \multicolumn{2}{c|}{} & \makecell[r]{0.17, 0.17, 0.17} & timeout & RB  \\\hline
\end{tabular}
\caption{Comparison with Verification Techniques used by Designers.}
\label{table:summary_results:conv}

\end{table*}

In this section, we present the statistics of the designs used in our experiments (Table~\ref{table:summary_design}) and comparison of our technique with conventional verification techniques (Table~\ref{table:summary_results:conv}).

In Table~\ref{table:summary_design}, we report the number of versions of each design, its source code type, \textbf{P}hysical \textbf{L}ines \textbf{O}f \textbf{C}ode, flip-flop counts (\textbf{FFs}) and gate counts (\textbf{Gates}). The average across different versions of a design are reported for each of the \textbf{Size} statistics. PLOC is calculated on the original source code using SLOCCount~\cite{slocc}. FFs and Gates are obtained from Cadence JasperGold reports when the RTL of the design (if design is expressed in a language other than RTL, the design is HLS synthesized using the HLS tools used by the designers) is elaborated in JasperGold.

In Table~\ref{table:summary_results:conv}, we compare our technique with conventional simulation based verification which the designers have used to verify each of these designs. The Setup \textbf{Effort} is reported in person-days (\textbf{PD}), which excludes CPU runtime. We see a $\sim$ 5X improvement in setup effort on the average with a $\sim$ 9X improvement for the large, industrial NVDLA designs. Avg. runtimes for \aqedd result from dividing the time to detect all bugs by the number of bugs. A verification running for over 12 hours was deemed \textbf{timeout}.  Avg. runtimes for conventional simulation result from dividing the runtime for each simulation trace by the total number of simulation traces. Some simulation statistics are not reported by the designers (\textbf{N/A}). \textbf{keypair}~\cite{pqc2019}, \textbf{gsm}~\cite{hara2008chstone}, \textbf{HLSCNN}~\cite{whatmough201916nm}, \textbf{FlexNLP}~\cite{Tambe_isscc2021}, \textbf{Dataflow}~\cite{chi2019rapid}, and \textbf{Opticalflow}~\cite{zhou2018rosetta} do not feature any sub-accelerators with a batch size greater than one. One OOB bug was found in \textbf{gsm} and one initialization bug was found in \textbf{keypair}. dRB Results are omitted for \textbf{nv\_large} and \textbf{nv\_small}; RB bugs result from parallelism and pipelining, both of which were lost in our manual translation of NVDLA from Verilog to C code. The bug type at the full design level was identified from understanding the bug after detection.

\end{appendices}


\end{document}